\documentclass{acm_proc_article-sp}
\usepackage{color,subfigure,url}

\definecolor{subtler}{rgb}{1,0,0.1}    

\begin{document}

\title{ Performance Evaluation of the Random Replacement Policy for Networks of  Caches}

\newtheorem{theorem}{Theorem}[section]
\newtheorem{lemma}[theorem]{Lemma}
\newtheorem{prop}[theorem]{Proposition}
\newtheorem{corol}[theorem]{Corollary}
\newtheorem{remark}{Remark}[section]

\def\rhoa{\rho_\alpha}
\def\ral{r^\alpha}
\def\Zipf{\frac{A}{r^\alpha}}

\def\ie{\emph{i.e.}}
\def\eg{\emph{e.g.}, }
\def\etal{\emph{et. al.}}

\author{Massimo Gallo, Bruno Kauffmann, Luca Muscariello,\\ 
 Alain Simonian, Christian Tanguy\\
\affaddr{Orange Labs, France Telecom}\\
\email{firstname.lastname@orange.com}}

\maketitle

\begin{abstract}
The overall performance of content distribution networks as well
as recently proposed information-centric networks rely on both
memory and bandwidth capacities. In this framework, the hit ratio
is the key performance indicator which captures the bandwidth /
memory tradeoff for a given global performance. 
\\
This paper focuses on the estimation of the hit ratio in a
network of caches that employ the Random replacement policy. 
Assuming that requests are independent and identically
distributed, general expressions of miss 
probabilities for a single Random cache are provided as well as exact results for specific popularity distributions.
Moreover, for any Zipf popularity distribution with exponent $\alpha > 1$, we obtain asymptotic equivalents for the miss probability in the case of large cache size.  \\
We extend the analysis to networks of Random caches, when the topology is either a line or a homogeneous tree. In that case, approximations for miss probabilities across the network are
derived by assuming that miss events at any node occur independently in time; the obtained results are compared to the same network using the Least-Recently-Used discipline,
already addressed in the literature. We further analyze the case of a mixed tandem cache network where the two nodes employ either Random or Least-Recently-Used policies. 
In all scenarios, asymptotic formulas and approximations are extensively 
compared to simulations and shown to perform very well. Finally, our results enable us to propose recommendations for cache replacement disciplines in a network dedicated to content
distribution. 
These results also hold for a cache using the First-In-First-Out policy.
\end{abstract}
 
\category{C.2.1}{Computer-Communication Networks}{Network Architecture and Design}[Packet-switching networks]

\terms{Performance Evaluation, Caching Networks}
\keywords{Content Distribution Networks, Information-Centric Networking, Cache Replacement Policies, Asymptotic Analysis}

\section{Introduction}
\label{sec:intro}
Communication networks use an ever increasing amount of data storage to cache information in transit from a source to a destination. Data caching is an auxiliary function, where a given
piece of data is temporarily stored into a memory.  
The cache might then be queried at any time to provide an object that is possibly stored in that memory. Caches are typically located downstream a bandwidth bottleneck, \eg{} a
communication link with limited bandwidth or a shared bus in a network of chips. This storage then allows to increase the throughput of the data path and to decrease the rate of congestion events. 

Many communication networks fall into such a model: the Internet,
Content Delivery Networks (CDNs), Peer-To-Peer (P2P) as well as
networks on chips. In fact, the increasing amount of content
delivered to Internet users has pushed the use of Web caching
into communication models based on distributed caching such as
CDNs. New information-centric network architectures
\cite{Jacobson2009Conext, Koponen:2007, sail} have been recently
proposed in order to have built-in network storage as a
fundamental element of the underlying communication model.
Content storage becomes a fundamental resource in such networks,
aiming at minimizing content delivery time under an ever increasing
demand that cannot simply be satisfied by increasing link bandwidth.
Moreover, the use of network storage, enabling one to cache content in order to bypass  bandwidth bottlenecks, appears cost effective as memory turns out to be cheaper
than transmission capacity.

One of the fundamental operations of a cache is defined by its replacement policy which   determines the object to be removed from the cache when the latter is full. 
Many replacement policies are based on content popularity, with
significant cost for managing the sorted lists. This is the case,
in particular, for the Least Frequently Used (LFU) policy and
more sophisticated variants based on it. On the contrary, Most
Recently Used (MRU), Least Recently Used (LRU),
First-In-First-Out (FIFO) and Random (RND)
policies have the compelling feature to replace cached objects
with constant delay $O(1)$. In-network storage, as envisaged in
the new architectures mentioned above, may require packet-level
caching at line rate; current  high-speed routers running complex
replacement policies might not, however, sustain such high line
rates \cite{varvelloICN2011}. 
In this framework, the RND and FIFO policies can therefore be
seen as presenting the least possible  complexity. In fact, RND
and FIFO requires 
less memory access per packet than LRU or MRU, and it has been shown \cite{varvelloICN2011} that this advantage is critical for sustaining high-speed caching with current memory technology. 

In this paper, we address performance issues of caching networks running the RND replacement policy. We mainly focus on the analytical characterization of the miss probabilities under the Independent Request Model (IRM) assumptions when the number of available objects is infinite. We first provide exact formula for the miss rate in the case of a light-tailed distribution content popularity, namely geometric, and for specific Zipf popularity distributions. Proposition~\ref{ASYMPT} then proves that when the popularity distribution follows a general power-law with decay exponent $\alpha > 1$, the miss probability is asymptotic to $A \rho_\alpha C^{1-\alpha}$ for large cache size $C$, where constants $A$ and $\rho_{\alpha}$ depend on $\alpha$ only. In Proposition~\ref{mrAsympt}, we extend that result to miss probabilities conditioned by the object popularity rank. 

A second major contribution of the paper is given by Proposition~\ref{APPROXNCRK}, where we evaluate the performance of network of caches under the RND policy, for both linear and homogeneous tree networks and asymptotically Zipf popularity distributions.
An approximate closed formula for the miss probability across the
network is provided and compared to corresponding estimates for
LRU cache networks. The analysis is also extended to the mixed
tandem cache network where one cache employs LRU and the other
one uses RND. Note that the choice of FIFO replacement policy
instead of RND does not impact any result.
 
The specific focus on Zipf distributions, or more generally power-law distributions, made in this paper is motivated by numerous studies on Internet object popularity, starting from the late 90's experiments on World Wide Web documents (\cite{Barford},\cite{Breslau}), the content stored in enterprises media servers (\cite{Cherkasova},\cite{Costa}) and to recent studies on Internet media content (\cite{Cha}, \cite{Mitra:2011:CWV:1961659.1961662}). While other content popularity distributions might be considered, we do not provide here a complete review of the literature on Internet content popularity characterization; the above references confirm the pertinence of Zipf distributions for studying caching performance.

The remainder of the paper is organized as follows. Section \ref{sec:related} presents related work on the analytical performance evaluation of caching systems. Section \ref{sec:single_cache} analyzes the RND cache replacement policy and its comparison to LRU for a single cache; these analytical results are compared with exact numerical and simulation results in Section~\ref{sec:num_single_cache}. Section~\ref{sec:network_cache}  
reports the approximate analysis of the network of RND caches for two topologies, namely the line and the tree. Numerical and simulation results for the network case are
reported in Section \ref{sec:num_network_cache}. Section \ref{sec:mixed} further evaluates the tandem cache system with one LRU and one RND cache, using numerical and simulation results. Section~\ref{sec:conclusion} concludes the paper (several proofs are detailed in the Appendix). 

\section{Related Work}
\label{sec:related}
There is a significant body of work on caching systems and their associated replacement policies that we will not attempt to review thoroughly. Here we only report the literature focusing on the analytical characterization of the performance of such systems.  

The most frequently analyzed replacement policy is LRU whose
performance is evaluated considering the move-to-front rule,
consisting in putting the latest requested object in front 
of a list; a miss event for a LRU cache with finite size takes
place when the position (also referred to as search cost) of an
object in the list is larger than that size. Under the
Independent Request Model, \cite{McCabe} calculates the expected
search cost and its variance for finite lists. An explicit
formula is given in \cite{Burville,Hendricks} for the probability
distribution of that cost; such a formula is, however,
impractical for numerical evaluation in case of large object
population and large cache size. Integral representations
obtained in \cite{FillHolst1996, Flajolet1992birthday} using the
Laplace transform of the search cost function reduce the problem
to numerical integration.  

An asymptotic analysis of LRU miss probabilities \cite{JelenkovicAAP1999} for Zipf and Weibull popularity distributions provides simple closed formulas. Extensions to correlated requests are obtained in \cite{Coffman1999MMPP, Jelenkovic2004DependentRequests},
showing that short term correlation does not impact the asymptotic results derived in \cite{JelenkovicAAP1999}; the case of variable object sizes is also considered in
\cite{Jelenkovic2004Size}. The analytical evaluation of the RND
policy has been first initiated by \cite{Gelenbe1973} for a
single cache where the miss probability is given by a general
expression. To the best of our knowledge, its application to
specific popularity distributions has, however, not yet been
envisaged together with its numerical tractability for large
object population and cache size.   

We are aware of few papers that address the issue of networks of caches.
A network of LRU caches has been analyzed by \cite{towsley2010approximate} using the  approximation for the miss probability at each node obtained in \cite{Dan1990}, and
assuming that the output of a LRU cache is also IRM; miss probabilities can then be obtained as the solution of an iterative algorithm that is proved to converge. The results of \cite{JelenkovicAAP1999} are extended in \cite{muscaITC2011} to a two-level request process where objects are segmented into packets, when assuming that the LRU policy applies to packets. The analysis is applied to line and tree topologies with in-path caching. Moreover, \cite{muscaICN2011} extends \cite{muscaITC2011} when network links have finite bandwidth. 

\section{Single cache model}\label{sec:single_cache}

In this section, we address the single cache system with RND replacement policy, deriving analytic expressions of the miss probability together with asymptotics for large cache size. To avoid technical results at first reading, the reader may quickly go through the notation of Section 3.1 and directly skip to main results given in Propositions \ref{ASYMPT} and \ref{mrAsympt}.

\subsection{Basic results}

Consider a cache memory with finite size which is offered requests for objects. If the cache size is attained and a request for an object cannot be satisfied (corresponding to a \textit{miss} event), the requested object is fetched from the repository server and cached locally at the expense of replacing some other object in the cache. The object replacement policy is assumed to follow the RND discipline, \ie whenever a miss occurs, the object to be replaced is chosen at random, uniformly among the objects present in cache. 

We consider a discrete time system: at any time $t \in \mathbb{N}$, the $t$-th requested object requested from the cache is denoted by $R(t) \in \{1, 2, ..., N\}$, where $N$ is the total 
number of objects which can be possibly requested (in the present
analysis, the set of all possible documents is considered to be
invariant in time). We assume that all $N$ objects are ordered
with decreasing popularity, the probability that object $r$ is
requested being $q_r$, $1 \leq r \leq N$. 
In the following, we consider the \textit{Independent
  Reference Model}, where  
variables $R(t)$, $t \in \mathbb{N}$, are  mutually
independent and identically distributed with common distribution
defined by 
\[
\mathbb{P}(R = r) = q_r, \; \; \; 1 \leq r \leq N.
\]

Let $C \leq N$ be the cache capacity and denote by $\mathcal{N}_C$ the set of all ordered subsets $\{j_1, ..., j_C\}$ with $1  \leq j_k  \leq N$, $k \in \{1, ..., C\}$, and $j_k < j_\ell$ for $k < \ell$. Define the cache state at time $t \in \mathbb{N}$ by the vector $\mathbf{S}(t)$; $\mathbf{S}(t)$ may equal any configuration $\mathbf{s} = \{j_1, ..., j_C\} \in \mathcal{N}_C$. In the following, we let 
\begin{equation}
G(C) = \sum_{\{j_1,...,j_C \} \in \mathcal{N}_C} \; q_{j_1}...q_{j_C}
\label{QCN}
\end{equation}
with $G(0) = 1$. Let $M(C)$ finally denote the stationary miss probability calculated over all possibly requested objects.  
\\
\indent
It is shown \cite{Gelenbe1973} that the cache configurations $(\mathbf{S}(t))_{t \in \mathbb{N}}$ define a reversible Markov process with stationary probability distribution given by 
\begin{equation}
\mathbb{P}(\mathbf{S} = \mathbf{s}) = \frac{1}{G(C)} \prod_{j \in \mathbf{s}} q_j, \; \; \mathbf{s} \in \mathcal{N}_C;
\label{1stcomp}
\end{equation}
moreover (\cite{Gelenbe1973}, Theorem 4), probability $M(C)$ equals
\begin{equation}
M(C) = \frac{\displaystyle \sum_{\{j_1,...,j_C \} \in \mathcal{N}_C} q_{j_1}...q_{j_C} \sum_{r \notin \{j_1, ..., j_C\}}q_r}{\displaystyle \sum_{\{j_1, ...,j_C\} \in \mathcal{N}_C} q_{j_1}...q_{j_C}}.
\label{CM}
\end{equation}
In \cite{Gelenbe1973}, it is also shown that the cache
configuration stationary probability distribution and miss rate
probability are identical in the case of a FIFO cache, hence our results
apply for the FIFO policy as well.

Expression (\ref{CM}) can be actually written in terms of normalizing constants $G(C)$ and $G(C+1)$ only; this will give formula (\ref{CM}) a compact form  suitable for the derivation of both exact and asymptotic expressions for $M(C)$. 

\begin{lemma}
The miss rate $M(C)$ is given by
\begin{equation}
M(C) = (C+1) \frac{G(C+1)}{G(C)}
\label{CMbis}
\end{equation}
with $G(C)$ defined in (\ref{QCN}).
\label{MAlter}
\end{lemma}

\begin{proof} Any state $\mathbf{s} = \{j_1, ..., j_C\} \in \mathcal{N}_C$ corresponding to a unique sequence $1 \leq j_1 < j_2 < ... < j_C \leq N$ with $1 \leq j_k \leq N$, the denominator of (\ref{CM}) therefore equals $G(C)$. The numerator of (\ref{CM}) can in turn be expressed as
\begin{align}
& \sum_{1 \leq j_1 < ... < j_C \leq N} q_{j_1}...q_{j_C} \sum_{r \notin \{j_1, ..., j_C\}}q_j 
\nonumber \\
& = \sum_{1 \leq j_1 < ... < j_C \leq N} q_{j_1}...q_{j_C} 
\nonumber \\
& \times \left ( \sum_{1 \leq r < j_1}q_r + ... + \sum_{j_{C-1} < r < j_C}q_r + \sum_{j_C < r \leq N} q_r \right )
\nonumber \\
& = (C+1) \sum_{1 \leq r < j_1 < ... < j_C \leq N} q_rq_{j_1}...q_{j_C} = (C+1)G(C+1)
\nonumber
\end{align}
and expression (\ref{CMbis}) of $M(C)$ follows
\end{proof} 

The latter results readily extend to the case when the total number $N$ of objects is infinite, since the series $\Sigma_{j \geq 1} q_j$ is finite. The calculation of coefficients $G(C)$, $0 \leq C \leq N$, is now performed through their associated generating function $F$ defined by
\[
F(z) = \sum_{0 \leq C \leq N} G(C) z^C, \; \; z \in \mathbb{C},
\]
for either finite or infinite population size $N$ (as $M(C) \leq 1$, Lemma \ref{MAlter} entails that $G(C+1)/G(C) \leq 1/(C+1)$ and the ratio test implies that the power series defining $F(z)$ has infinite convergence radius). We easily obtain the second preliminary result.

\begin{lemma}
The generating function $F$ is given by 
\begin{equation}
F(z) = \prod_{1 \leq r \leq N} \left ( 1 + q_r z \right )
\label{genf}
\end{equation}
for all $z \in \mathbb{C}$.
\label{GEN}
\end{lemma}

\begin{proof} Expanding the latter product and using definition (\ref{QCN}) readily provide the result 
\end{proof}

To further study the single cache properties, let $M_r(C)$ denote the per-object miss probability, given that the requested object is precisely $r \in \{1, ..., N\}$. Defining
\begin{equation}
G_r(C) = \sum_{1 \leq j_1 < ... < j_C \leq N, \; r \notin \{j_1, ...,j_C\}} \; q_{j_1}...q_{j_C}
\label{QCNr}
\end{equation}
with $G_r(0) = 1$, we then have $M_r(C) = \mathbb{P}(r \notin \mathbf{S})$ so that (\ref{1stcomp}) and (\ref{QCNr}) yield
\begin{equation}
M_r(C) = \frac{G_r(C)}{G(C)}.
\label{mr}
\end{equation}

\begin{lemma}
The per-object miss probability $M_r(C)$ for given $r \in \{1, ..., N\}$ can be expressed as
\begin{equation}
M_r(C) = 1 + \sum_{\ell = 0}^{C-1} (-q_r)^{C-\ell} \frac{G(\ell)}{G(C)}.
\label{mrbis}
\end{equation}
The stationary probability $q_r(2)$, $r \in \{1, ..., N\}$, that a miss event occurs for object $r$ is given by
\begin{equation}
q_r(2) = \frac{M_r(C)}{M(C)}q_r
\label{dismissr}
\end{equation}
where $M(C)$ is the averaged miss probability.
\label{GENr}
\end{lemma}

\begin{proof} By definition (\ref{QCNr}), the generating function $F_r(z)$ of coefficients $G_r(C)$, $0 \leq C \leq N$, is given by
\begin{equation}
F_r(z) = \frac{F(z)}{1+q_rz}, \; \; z \in \mathbb{C}.
\label{Fr}
\end{equation}
Expanding the latter ratio as a powers series of $z$ gives
$$
G_r(C) = \sum_{\ell = 0}^C (-q_r)^{C-\ell} G(\ell) 
$$
and provides (\ref{mrbis}) after using definition (\ref{mr}) for $M_r(C)$. Besides, letting  $\mathcal{M}$ denote a miss event, Bayes formula entails
\begin{align}
q_r(2) & = \mathbb{P}(R=r \; |\; \mathcal{M}) =
\mathbb{P}(R=r)\frac{\mathbb{P}(\mathcal{M} \; | \;
  R=r)}{\mathbb{P}(\mathcal{M})}
 = q_r \frac{M_r(C)}{M(C)}
\nonumber
\end{align}
hence relation (\ref{dismissr}) 
\end{proof}

If the popularity distribution has unbounded support, then $\lim_{r \uparrow +\infty} q_r = 0$ and formula (\ref{mrbis}) implies that the per-object probability $M_r(C)$ tends to 1 as $r \uparrow +\infty$ for fixed $C$; (\ref{dismissr}) consequently entails 
\begin{equation}
q_r(2) \sim \frac{q_r}{M(C)}, \; \; \; r \uparrow +\infty.
\label{asymptqtilde}
\end{equation}
For given $C$, asymptotic (\ref{asymptqtilde}) shows that the tail of distribution $(q_r(2))_{r \in \mathbb{N}}$ at infinity is proportional to that of distribution $(q_r)_{r \in \mathbb{N}}$. The distribution $(q_r(2))_{r \in \mathbb{N}}$ describes the output process of the single cache generated by consecutive missed requests; it will serve as an essential ingredient to the further extension of the single cache model to network cache configurations considered in Section \ref{sec:network_cache}. 

\subsection{First applications}

Coefficients $G(C)$, $C \geq 0$, and associated miss probability $M(C)$ can be explicitly derived for some specific popularity distributions. In the following, the total population $N$ of objects is always assumed to be infinite. 

\begin{corol}
Assume a geometric popularity distribution $q_r =
(1-\kappa)\kappa^{r-1}$, $r \geq 1$, with given $\kappa \in \;
]0,1[$. For all $C \geq 0$, the miss rate equals 
\begin{equation}
M(C) = \frac{1-\kappa}{1-\kappa^{C+1}}(C+1)\kappa^C .
\label{geometric}
\end{equation}

\label{GEO}
\end{corol}

\begin{proof} Using Lemma \ref{GEN}, $F$ is readily shown to
  verify the functional identity 
$F(z) = (1 + (1-\kappa)z) F(\kappa z)$
for all $z \in \mathbb{C}$. Expanding each side of that identity in power series of $z$ and identifying identical powers provides the value of the ratio $G(C+1)/G(C)$, hence result (\ref{geometric}) by (\ref{CMbis})
\end{proof}

Let us now assume that the popularity distribution follows a Zipf distribution defined by
\begin{equation}
q_r = \frac{A}{r^\alpha}, \; \; r \geq 1,
\label{Zipf}
\end{equation}
with exponent $\alpha > 1$ and  normalization constant $A = 1/\zeta(\alpha)$, where $\zeta$ is the Riemann's Zeta function. We now show how explicit rational expressions for miss rate $M(C)$ can be obtained for some integer values of $\alpha$.

\begin{corol}
Assume a Zipf popularity distribution with exponent $\alpha$. 
For all $C \geq 0$, the miss probability equals
\begin{align}
M(C) & = \displaystyle \frac{3}{2C+3} & \mathrm{if} & \; \; \alpha = 2,
\nonumber \\
& = \displaystyle \frac{45}{(4C+5)(4C+3)(2C+3)} & \mathrm{if} & \; \; \alpha = 4,
\nonumber \\
& = \displaystyle \frac{9!}{3!} \frac{ (C+1)}{\displaystyle \prod_{4 \leq j \leq 9}{(6C + j)}} & \mathrm{if} & \; \; \alpha = 6,
\nonumber
\end{align}
\label{Zipf24}
\end{corol}

\begin{proof} When $\alpha = 2$, the normalization constant equals $A = 1/\zeta(2) = 6/\pi^2$. From the infinite product formula \cite{Abram}
$$
F(z) = \prod_{j \geq 1} \left ( 1 + \frac{u^2}{\pi^2 j^2} \right ) = \frac{\sinh u}{u}
$$
and expanding the left hand side into powers of $u^2 = Az \pi^2$ gives the expansion $F(z) = \Sigma_{C \geq 1} G(C) z^C$ where
$$
G(C) = (\pi^2 A)^C/(2C + 1)!, \; \; C \geq 0.
$$
Computing then ratio (\ref{CMbis}) with the above expression of $G(C)$ then provides $M(C) = 3/(2C+3)$, as claimed. The cases when $\alpha = 4$ or $\alpha = 6$ follow a similar derivation pattern 
\end{proof}
The formulas of Corollary \ref{Zipf24} do not seem, however, to generalize for integer values $\alpha = 2p$ with $p \geq 4$; upper bounds can be envisaged and are the object of further study.

As also suggested by Corollary \ref{Zipf24}, the cache size corresponding to a target miss probability should be a decreasing function of $\alpha$. 
This property is generalized in \textbf{Section \ref{sec:large_cache}} where an asymptotic evaluation of $M(C)$ is provided for large cache size $C$ and any Zipf popularity distribution with real exponent $\alpha > 1$.

\subsection{Large cache approximation} \label{sec:large_cache}

The specific popularity distributions considered in Corollaries \ref{GEO} and \ref{Zipf24} show that $M(C)$ is of order $C q_C$ for $C$. In the present section, we derive general  asymptotics for probabilities $M(C)$ and $M_r(C)$ for large cache size and apply them to the Zipf popularity distribution. 

We first start by formulating a general large deviations result (Theorem \ref{LD}) for evaluating coefficients $G(C)$ for large $C$. To apply the latter theorem to the Zipf distribution (\ref{Zipf}), we then state two preliminary results (Lemmas \ref{ASYMPTF} and \ref{SADDLEPOINTC}) on the behavior of the corresponding generating function $F$ at infinity. This finally enables us to claim our central result (Proposition \ref{ASYMPT}) for the behavior of $M(C)$ for large $C$.

\begin{theorem} (See Proof in \textbf{Appendix \ref{LargeDeviations}})
\\
(i) Given the generating function $F$ defined in (\ref{genf}), equation
\begin{equation}
z  F'(z) = C F(z)
\label{saddlepoint0}
\end{equation}
has a unique real positive solution $z = \theta_C$ for any given $C \geq 0$;
\\
\indent
(ii) Assume that there exists some constant $\sigma > 0$ such that the limit
\begin{equation}
\lim_{C \uparrow +\infty} e^{s\sqrt{C}} \frac{F(\theta_C e^{-s/\sqrt{C}})}{F(\theta_C)} = e^{\sigma^2 s^2/2}
\label{Gaussapprox1}
\end{equation}
holds for any given $s \in \mathbb{C}$ with $\Re(s) = 0$ and that, given any $\delta > 0$, there exists $\eta \in ]0, 1[$ and an integer $C_\delta$ such that
\begin{equation}
\sup_{\delta \leq \mid y \mid \; \leq \pi} \left |  \frac{F(\theta_C e^{iy})}{F(\theta_C)} \right | ^{1/C} \leq \eta
\label{Gaussapprox2}
\end{equation}
for $C \geq C_\delta$. We then have
\begin{equation}
G(C) \sim \frac{\exp(H_C)}{\sigma \sqrt{2\pi C}}
\label{Saddlepointfin}
\end{equation}
as $C$ tends to infinity, with $H_C = \log F(\theta_C) - C \log \theta_C$.
\label{LD}
\end{theorem}

Following Theorem \ref{LD}, the asymptotic behavior of $M(C)$ can then be derived from (\ref{Saddlepointfin}) together with identity (\ref{CMbis}). This approach is now applied to the Zipf popularity distribution (\ref{Zipf}); in this aim, the behavior of the corresponding generating function $F$ is first specified as follows.

\begin{lemma} (See Proof in \textbf{Appendix \ref{ProofAsymptF}})
\\
For $\alpha > 1$ and large $z \in \mathbb{C} \setminus \mathbb{R}^-$, $\log F(z)$ expands  as 
\begin{equation}
\log F(z) = \alpha (\rho_\alpha Az)^{1/\alpha} - \frac{1}{2} \log(Az) + S_\alpha + o(1)
\label{asympt}
\end{equation}
with $A = 1/\zeta(\alpha)$ and $S_\alpha$ depending on $\alpha$
only and 
\begin{equation}
\rho_\alpha = \left ( \frac{\pi/\alpha}{\sin(\pi/\alpha)} \right ) ^\alpha, 
\label{prefP}
\end{equation}
 
\label{ASYMPTF}
\end{lemma}

\begin{lemma} (See Proof in \textbf{Appendix \ref{ProofSaddlepointC}})
\\
For $\alpha > 1$ and large $C$, the unique real positive solution $\theta_C$ of equation (\ref{saddlepoint0}) verifies
\begin{equation}
\theta_C = \frac{C^\alpha}{A\rho_\alpha} + C^{\alpha - 1} r_C
\label{zC}
\end{equation}
with $r_C = A_1 + O(C^{-1})$ with $A_1 = \alpha/2 \rho_\alpha A$ if $\alpha \neq 2$ and $r_C = O(\log C)$ if $\alpha = 2$.
\label{SADDLEPOINTC}
\end{lemma}

We can now state our central result.

\begin{prop}
For a Zipf popularity distribution with exponent $\alpha > 1$, the miss probability $M(C)$ is asymptotic to
\begin{equation}
M(C) \sim \rho_\alpha C q_C = \frac{A \rho_\alpha}{C^{\alpha -1}}
\label{asymptfin}
\end{equation}
for large $C$, with prefactor $\rho_\alpha$ given in (\ref{prefP}). 
\\
Prefactor $\rho_\alpha$ verifies $\lim_{\alpha \uparrow +\infty} \rho_\alpha = 1$ and $\rho_\alpha \sim 1/(\alpha - 1)$ as $\alpha \downarrow 1$.
\label{ASYMPT}
\end{prop}

\begin{proof} As verified in \textbf{\textbf{Appendix \ref{ProofAsympt}}}, the conditions of Theorem \ref{LD} are satisfied for a Zipf distribution. Using asymptotics  (\ref{asympt}) and (\ref{zC}) of Lemmas \ref{ASYMPTF} and \ref{SADDLEPOINTC} to explicit the argument $H_C$ in (\ref{Saddlepointfin}) for large $C$, we then have
\begin{align}
H_C & = \log F(\theta_C) - C \log \theta_C 
\nonumber \\
& = \alpha(\rho_\alpha A\theta_C)^{1/\alpha} - \frac{1}{2}\log(A \theta_C) + S_\alpha + o(1) -
\nonumber \\
& \; \; \; \; C \log \left [ \frac{C^\alpha}{A\rho_\alpha} + C^{\alpha - 1}r_C \right ]
\nonumber
\end{align}
so that $H_{C+1} - H_C = -\alpha \log C + k + o (1)$ with
constant $k = \log \rho_\alpha + \log A$. Coming back to
definition (\ref{CMbis}) of $M(C)$, the latter estimates enable
us to derive that
\begin{align}
M(C) & = (C+1) \frac{G(C+1)}{G(C)} \sim C \exp [ H_{C+1} - H_C ] 
\nonumber \\
& \sim C \frac{e^k}{C^{\alpha}} = \rho_\alpha \; C \; q_C 
\end{align}
as claimed\end{proof} 

\begin{remark} Proposition \ref{ASYMPT} provides asymptotic (\ref{asymptfin}) for $M(C)$ under the weaker assumption that the popularity distribution $q_r$, $r \geq 1$, has a heavy tail of order $r^{-\alpha}$ for large $r$ and some $\alpha > 1$, without being precisely Zipf as in (\ref{Zipf}). In fact, all necessary properties for deriving Lemmas \ref{ASYMPTF} and \ref{SADDLEPOINTC} are based on that tail behavior only.
\label{tailonly}
\end{remark}

To close this section, we now address the asymptotic behavior of $M_r(C)$ defined in (\ref{mr}). 

\begin{prop}
For any Zipf popularity distribution with exponent $\alpha > 1$ and given the object rank $r$, the per-object miss probability $M_r(C)$ is estimated by
\begin{equation}
M_r(C) \sim \frac{\rho_\alpha r^\alpha}{C^\alpha + \rho_\alpha r^\alpha}
\label{asymptmrC}
\end{equation}
for large $C$, with prefactor $\rho_\alpha$ defined in (\ref{prefP}).
\label{mrAsympt}
\end{prop}

\begin{proof} The generating function $F_r$ of the sequence $G_r(C)$, $C \geq 0$, being given by (\ref{Fr}), apply Theorem \ref{LD} to estimate coefficients $G_r(C)$ for large $C$. Concerning condition \textit{(i)}, the solution $\eta = \eta_C$ to equation $\eta F'_r(\eta) = C F_r(\eta)$ reduces to equation (\ref{saddlepoint0}) for $\theta = \theta_C$ where the term $q_r \theta/(1+q_r\theta)$ has been suppressed; but suppressing that term does not modify the estimate for $\theta_C$ with large $C$ so that $\eta_C \sim \theta_C$. On the other hand, condition \textit{(ii)} is readily verified by generating function $F_r$ and we then obtain 
\begin{equation}
G_r(C) \sim \frac{G(C)}{1+q_r \theta_C}.
\label{asymptmrC1}
\end{equation}
By Lemma \ref{SADDLEPOINTC}, we have $\theta_C \sim C^\alpha/A\rho_\alpha$ for large $C$; definition (\ref{mr}) of $M_r(C)$ and estimate (\ref{asymptmrC1}) with  $q_r = A/r^\alpha$ give 
\begin{equation}
M_r(C) = \frac{G_r(C)}{G(C)} \sim \frac{1}{1+q_r \theta_C} \sim
\frac{\rhoa \ral}{C^\alpha + \rhoa \ral}
\label{asymptmrC2}
\end{equation}
and result (\ref{asymptmrC}) follows
\end{proof}

For any value $\alpha > 1$, (\ref{asymptmrC}) is consistent with the fact that $M_r(C)$ is an increasing function of object rank $r$ and a decreasing function of cache size $C$.

\subsection{Comparing RND to LRU}

Let us now compare the latter results with the LRU replacement policy investigated in \cite{Flajolet1992birthday, JelenkovicAAP1999}. Recall that, for a Zipf popularity distribution with exponent $\alpha > 1$, the miss probability $M(C)$ for LRU is estimated by
$$
M(C) \sim \lambda_\alpha Cq_C
$$
for large $C$, with prefactor 
\begin{equation}
\lambda_\alpha = \frac{1}{\alpha}\left [ \Gamma \left ( 1 - \frac{1}{\alpha} \right ) \right ]^\alpha 
\label{prefP'}
\end{equation}
where $\Gamma$ is Gamma function (\cite{JelenkovicAAP1999}, Theorem 3). $\lambda_\alpha$ is estimated by 
$$
\lambda_\alpha \sim \frac{e^\gamma}{\alpha}, \; \; \; \; \lambda_\alpha \sim \dfrac{1}{\alpha - 1}
$$
as $\alpha \uparrow +\infty$ and $\alpha \downarrow 1$, respectively ($\gamma$ is Euler's constant and $e^\gamma \approx 1,781...$). In view of Proposition \ref{ASYMPT}, comparing asymptotics for coefficients $\lambda_\alpha$ and $\rho_\alpha$ shows that the difference $\rho_\alpha - \lambda_\alpha$ tends to 1 as $\alpha \uparrow
+\infty$, thus illustrating that LRU discipline performs better
than RND for large enough $\alpha$ (it can be formally shown that
$\rho_\alpha > \lambda_\alpha$ for all $\alpha > 1$). 
This difference diminishes, however, for smaller values of
$\alpha$ since $\rho_\alpha$ and $\lambda_\alpha$ behave
similarly as $\alpha$ is close to 1 (see Figure
\ref{Fig2}). Apart from that limited discrepancy, both
disciplines provide essentially similar performance levels in
terms of miss probabilities for large cache sizes.  

\begin{figure}[h!]
\begin{center}
\includegraphics[width=0.42\textwidth]{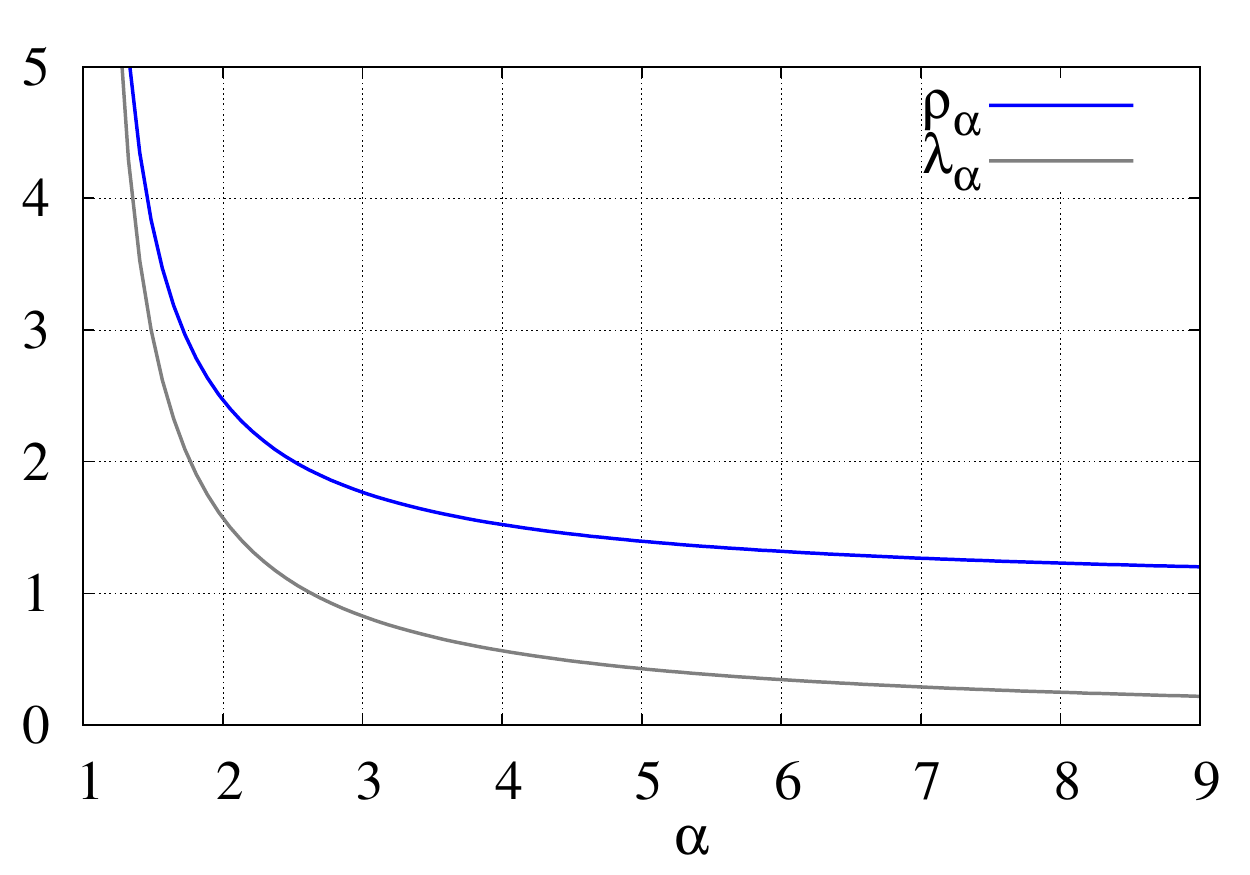}
\caption{\textit{Prefactors $\rho_\alpha$ and $\lambda_\alpha$ with $\alpha > 1$, for RND and LRU policies, respectively.}}
\label{Fig2}
\end{center}
\end{figure}

In contrast to heavy-tailed popularity distributions, we can consider a light-tailed distribution where
$$
q_r = A \exp(-B r^\beta), \; \; r \geq 1,
$$
with positive parameters $A, B, \beta$. It is shown in this case \cite{JelenkovicAAP1999} that the miss probability for LRU is asymptotic to
$$
M(C) \sim \frac{e^\gamma}{\beta B}C^{1-\beta}q_C
$$
for large $C$. For a geometric popularity distribution (with $\beta = 1$), the latter estimate shows that $M(C) = O(q_C)$; on the other hand, formula (\ref{geometric}) of Corollary \ref{GEO} shows that $M(C) = O(Cq_C)$ for RND discipline. This illustrates the fact that RND and LRU replacements provide significantly different performance levels if the popularity distribution is highly concentrated on a relatively small number of objects.

\section{Numerical results: single cache}
\label{sec:num_single_cache}

We here present numerical and simulation results to validate the
preceding estimates for a single RND policy cache. In the
following, when considering a finite object population with total
size $N$, the Zipf popularity distribution is normalized
accordingly. 
We also mention that the content popularity
distribution obviously refers to document classes instead of
individual documents. For comparison purpose with the existing
LRU analysis, we represent these classes by a single index, as
if they were a single document. In the following, cache sizes must
accordingly be scaled up to the typical class size.

Simulations are performed using an ad-hoc simulator written in
C. In every simulation, performance measures are collected after
the transient phase, once the system has reached the stationary
state. In this paper, transient behavior is  not considered; note
that the duration of the transient period obviously
increases with the cache size.  

\begin{figure}[!tb]
\begin{center}
\subfigure[]{\includegraphics[width=0.43\textwidth]{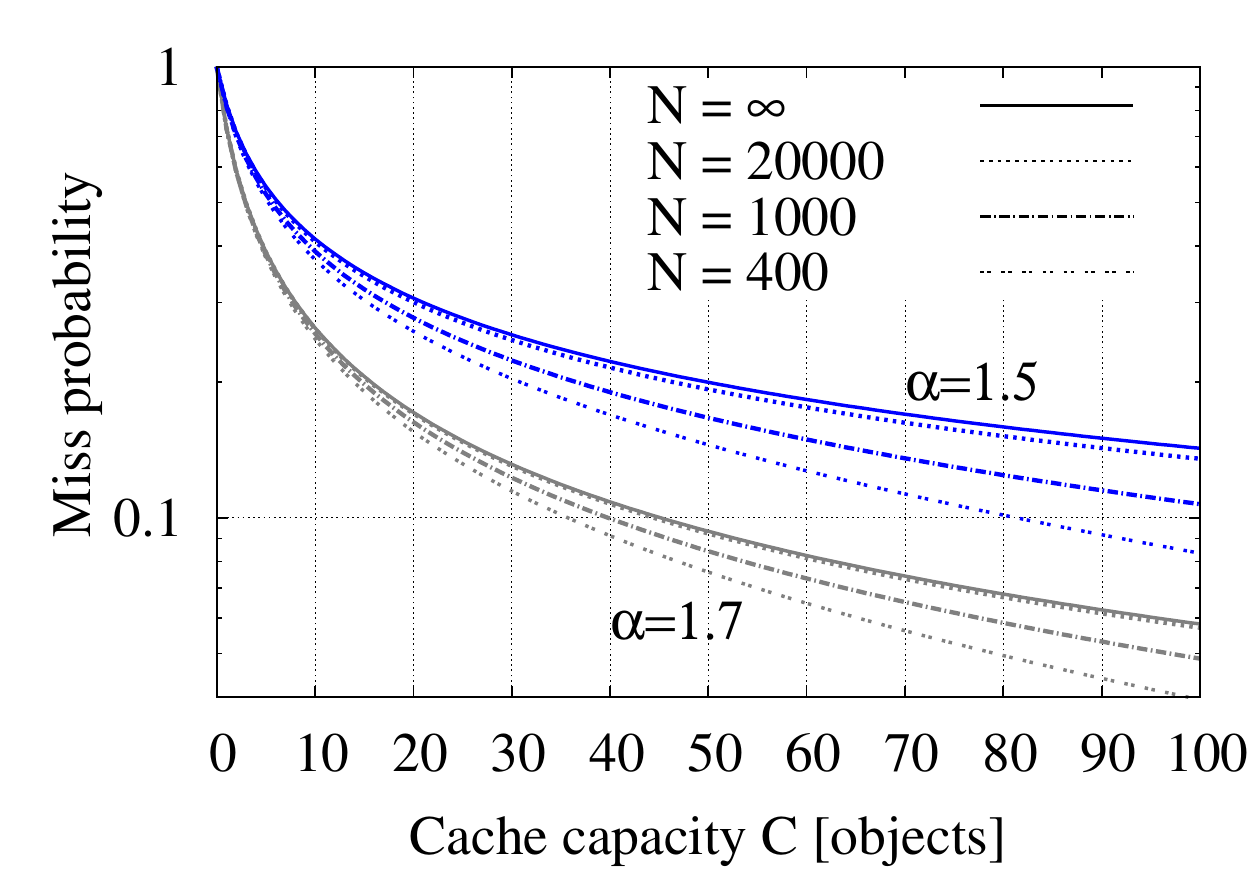}\label{fig:N_effect}}
\subfigure[]{\includegraphics[width=0.43\textwidth]{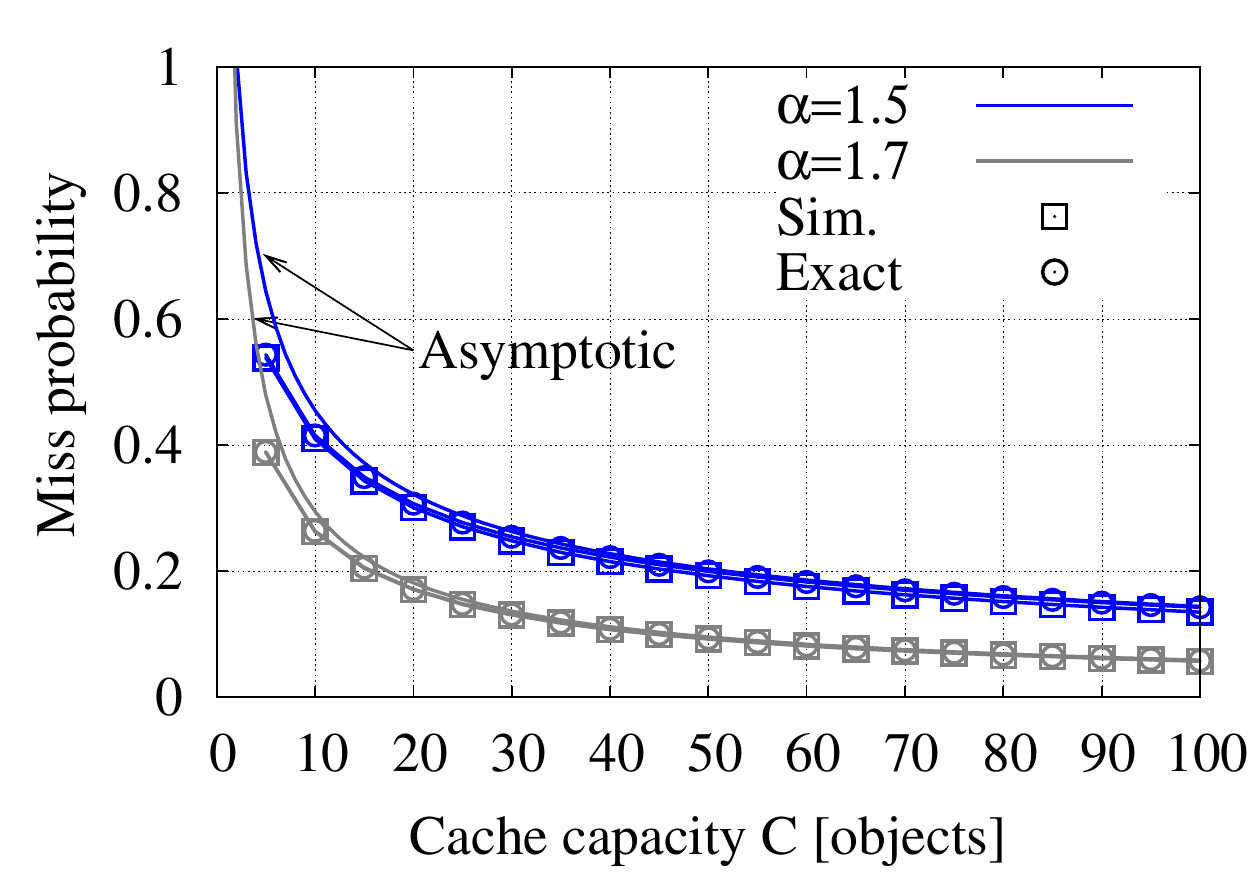}\label{fig:M_C}}
\subfigure[]{\includegraphics[width=0.43\textwidth]{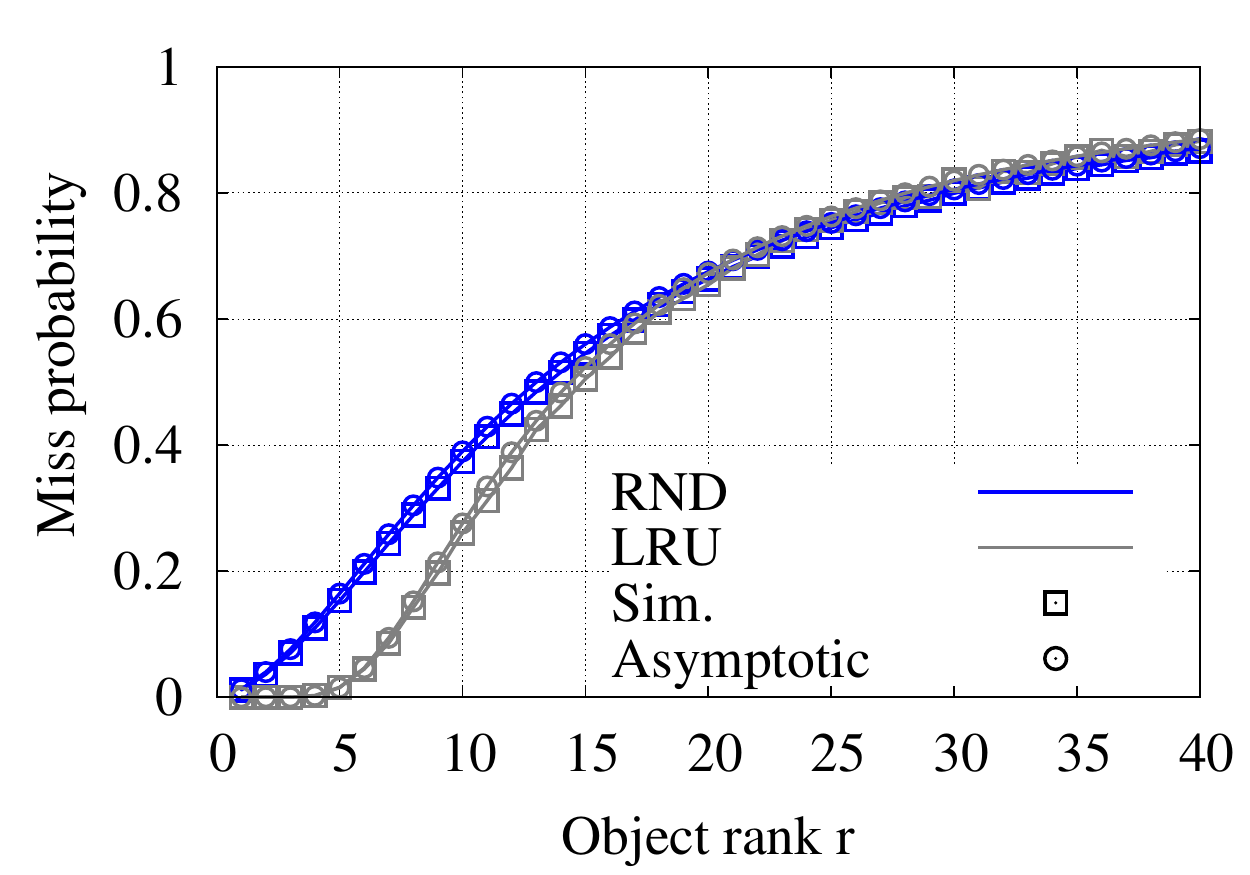}\label{fig:Mr_C25}}
\caption{Single cache results: (a) exact formula for $M(C)$ with RND policy (b) asymptotic  of $M(C)$ with RND policy (c) asymptotic of $M_r(C)$ with $C=25$, $\alpha=1.7$ for RND and LRU policies.}
\end{center}
\end{figure}

Besides, the most critical parameter in our simulation setting is
the numerical value of $\alpha$. 
As the Zipf distribution flattens when $\alpha$ get closer to 1, 
much longer simulation runs are necessary to have good estimates of
the miss rate. 
Small enough values of $\alpha$
must, nevertheless, be considered as they are more realistic. Estimates of $\alpha$ are reported, in particular, in
\cite{Mitra:2011:CWV:1961659.1961662} for web sites providing
access to video content like \url{www.metacafe.com} for which
$\alpha=1.43$, \url{www.dailymotion.com} and \url{www.veoh.com}
for which $\alpha=1.72$ and $\alpha = 1.76$, respectively. In the
following, we hence fix $\alpha = 1.5$ or $\alpha = 1.7$ in
our numerical experiments.

Fig.~\ref{fig:N_effect} first reports exact formula (\ref{CM})
for $M(C) = M(C;N)$ as a function of cache size $C$ and for
increasing values of total population $N$, where $M(C;N)$
measures the total miss probability for a cache of size $C$ when
the number of objects $N$ is finite. As expected, the
convergence speed of $M(C;N)$ to $M(C;\infty)$ as  $N \uparrow
+\infty$ increases with $\alpha$. In the case $\alpha=1.5$ for
instance, a population of $N=20 \; 000$ can be considered a good
approximation for an infinite object population
($N=\infty$), whilst there is almost no difference between
$M(C;N=20 \; 000)$ and $M(C;N=\infty)$ when $\alpha=1.7$.

In Fig.~\ref{fig:M_C}, we compare exact formula (\ref{CM}) for
$M(C)$, asymptotic (\ref{asymptfin}) and simulation results for
the above scenario. Formula (\ref{CM}) for $N=\infty$ is 
computed with arbitrary precision and we used $N=20 \; 000$ for
simulation as a good approximation for an infinite object
population. Simulation and exact results are very close
(especially for $\alpha=1.7$) while asymptotic (\ref{asymptfin})
gives a very good estimation of the miss probability as soon as
$C \geq 20$. 

Fig.~\ref{fig:Mr_C25} presents the miss probability $M_r(C)$ as a
function of the object rank $r$ for both RND and LRU policies
with fixed $C = 25$ and $\alpha=1.7$. Results are reported for
the most popular classes and confirm the asymptotic accuracy of
estimate (\ref{asymptmrC}) for RND and the corresponding one for
LRU policy \cite{Jelenkovic2008PER}. Beside the good
approximation provided by the asymptotics, it is important to
remark that RND and LRU performance are very close when object
rank $r \geq 15$, while there is a slight difference for the most
popular objects (say $r < 15$). Moreover, comparing $M(C=25)$ for
RND and LRU (respectively equal to 0.147 and 0.108), we observe
only 4\% less of miss probability using LRU with respect to RND
policy. This may suggest RND as a good candidate for caches
working at very high speed, where LRU may become too expensive in
terms of computation due to its relative complexity.  

\section{In-network cache model}
\label{sec:network_cache} 

In order to generalize the single-cache model, networks of caches with various topologies  can be considered.

\subsection{Line topology}
We first consider the tandem system defined as follows. Any request is addressed to a first cache $\sharp 1$ with size $C_1$; if it is not satisfied, it is addressed to a second cache $\sharp 2$ with size $C_2$:
\\
\indent 
- if this request is satisfied at cache $\sharp 2$, the object is copied to cache $\sharp 1$, with replacement performed according to the RND discipline; 
\\
\indent
- if this request is not satisfied at cache $\sharp 2$, the object is retrieved from a repository server and copied in caches $\sharp 1$ AND $\sharp 2$ according to the RND discipline. 
\\
\indent
Note that this replacement scheme, hereafter denoted by \textbf{IPC} for \textbf{In-Path Caching}, ignores any collaboration between the two caches and blindly copies objects in all caches along the path towards the requesting source. 
\\
\indent
We now fix some notation and properties for this tandem model. Let $R_1(t) \in \{1, 2, ..., N\}$ denote the object requested at cache $\sharp 1$ at time $t$; we still assume that variables $R_1(t)$, $t \in \mathbb{N}$, describe an IRM process with distribution defined by $\mathbb{P}(R_1 = r) = q_r$, $1 \leq r \leq N$. Denoting by $\mathbf{S}_1(t)$ (resp. $\mathbf{S}_2(t)$) the state vector of cache $\sharp 1$ (resp. $\sharp 2$) at time $t$, the bivariate process $(\mathbf{S}_1(t), \mathbf{S}_2(t))_{t \in \mathbb{N}}$ is easily shown to define a Markov process that, however, is not reversible. 
It is therefore unlikely that an exact closed form for the stationary distribution of process $(\mathbf{S}_1,\mathbf{S}_2)$ can be derived to evaluate the miss probability for the two-cache system. 
\\
\indent
Alternatively, we here follow an approach based on the approximation of the request process to cache $\sharp 2$. Let $t_n$, $n \in \mathbb{N}$, denote the successive instants when a miss occurs at first cache $\sharp 1$, and $R_2(n)$ be the object corresponding to that miss event at time $t_n$. First note that the common distribution of variables $R_2(n)$ is the stationary distribution $(q_r(2))_{r \in \mathbb{N}}$ introduced in Lemma \ref{GENr}, equation (\ref{dismissr}), with cache size $C$ replaced by $C_1$. In the following, we will further assume that 

\begin{itemize}
\item[\textbf{(H)}] the request process for cache $\sharp 2$ is an IRM, that is, all variables $R_2(n)$, $n \in \mathbb{N}$, are independent with common distribution
$$
\mathbb{P}(R_2 = r) = q_r(2), \; \; r \in \mathbb{N}.
$$
\end{itemize}

The simplifying assumption \textbf{(H)} neglects any correlation structure for the output process of cache $\sharp 1$ (that is, the input to cache $\sharp 2$) produced by consecutive missed requests. Recall also that the tail of distribution $(q_r(2))_{r \in \mathbb{N}}$, defined by (\ref{asymptqtilde}), is proportional to that of distribution $(q_r)_{r \in \mathbb{N}}$.

The latter 2-stage tandem model can be easily extended to a tandem network consisting in a series of $K$ caches ($K > 2$) where any request dismissed at caches $\sharp 1, ..., \sharp \ell$, $\ell \geq 1$, is addressed to cache $\sharp (\ell + 1)$. As an immediate generalization of the \textbf{IPC} scheme, we assume that any requested document which experiences a miss at cache $\sharp j$, $1 \leq j \leq \ell$, and an object hit at cache $\sharp (\ell+1)$ is copied backwards at all downstream caches $\sharp 1, ..., \sharp \ell$. A request miss therefore corresponds to a miss event at each cache $1, 2, ..., K$.  Furthermore, assumption \textbf{(H)} is generalized by saying that any cache $\sharp \ell$ considered in isolation behaves as a single cache with IRM input produced by consecutive missed requests at cache $\sharp (\ell -1)$. The size of cache $\sharp \ell$ is denoted by $C_\ell$. 

In the following, the "global" miss probability  $M_r(C_1,\ldots, C_\ell)$ (resp. "local" miss probability $M^*_r(C_1, \ldots, C_\ell)$) for request $r$ at cache $\ell$ is the miss probability for object $r$ over all caches $1, ..., \ell$ (resp. the miss probability for object $r$ at cache $\ell$) so that 
\begin{equation}
M_r(C_1, \ldots, C_\ell) = \prod_{j=1}^{\ell} M^*_r(C_1, \ldots, C_j)
\label{independence}
\end{equation} 
(note that for a single cache, we have $M_r(C_1) =
M_r^{*}(C_1)$). To simplify notation, we abusively write $M_r(\ell)$ (resp. $M^*_r(\ell)$) instead of $M_r(C_1, \ldots, C_\ell)$ (resp. instead of $M^*_r(C_1, \ldots, C_\ell)$). Finally, if $q_r(\ell)$, $r \geq 1$, defines the distribution of the input process at cache $\sharp \ell$, the averaged local miss probability $M^*(\ell)$ at cache $\sharp \ell$ is given by
\begin{equation}
M^*(\ell) = \sum_{r \geq 1} M_r^*(\ell)q_r(\ell)
\label{averagelocal}
\end{equation}
for any $\ell \in \{1, ..., K \}$.

\begin{prop} (See Proof in \textbf{Appendix \ref{ProofApproxNCRk}})
\\
For the $K$-caches tandem system with \textbf{IPC} scheme, suppose that the request process at cache $\sharp 1$ is IRM with Zipf popularity distribution with exponent $\alpha > 1$, and that assumption \textbf{(H)} holds for all caches $\sharp 2, ..., \sharp K$. 
\\
\indent
For any $\ell \in \{1, ..., K \}$ and large cache sizes $C_1, ..., C_\ell$, the global miss probability $M_r(\ell)$ (resp. local miss probability $M^*_r(\ell)$) is given by
\begin{equation}
M_r(\ell) \sim \frac{\rho_\alpha r^\alpha}{\displaystyle \rho_\alpha r^\alpha + \sum_{j=1}^\ell C_j^\alpha}, \; \; \; M^*_r(\ell) \sim \frac{\displaystyle \rho_\alpha r^\alpha + \sum_{j=1}^{\ell-1} C_j^\alpha}{\displaystyle \rho_\alpha r^\alpha + \sum_{j=1}^\ell C_j^\alpha}.
\label{Miss-r-NCRk}
\end{equation}
\label{APPROXNCRK}
\end{prop}

Proposition \ref{APPROXNCRK} shows how the $K$-stage tandem system with \textbf{IPC} scheme improves the performance in terms of miss probability by adding a term $C_j^\alpha$ when  the $j$-th cache is added to the path. From Propositions \ref{APPROXNCRK} and \ref{ASYMPT}, it is readily derived that the average global miss probability $M(\ell)$ for all objects requested along the cache network is
\begin{equation}
M(\ell) = \sum_{r \geq 1} M_r(\ell) q_r \sim
\frac{A \rhoa}
{\left(\displaystyle \sum_{j=1}^\ell C_\ell^\alpha\right)^{1-\frac{1}{\alpha}}}
\label{MissNCRk}
\end{equation}
for any $\ell \in \{1, ..., K \}$ and large cache sizes $C_1$, ..., $C_\ell$.

\subsection{Tree topology} 
\label{TreeSection}

The previous linear network model can be easily extended to the homogeneous tree topology with Zipf distributed requests. By homogeneous, we mean that all leaves of the tree are located at a common depth of the root, and that the cache size for each node at a given level $i$ is equal to $C_i$ (where $C_1$ is the cache size of the leaves). An example of such a tree is a complete binary tree of given height.

Let $\Lambda_1, \ldots, \Lambda_J$ be the $J$ leaves of the
tree. We assume that all requests arrive at the leaves, following
an IRM, that is, 
$\mathbb{P}(R(t) = r, \Lambda(t)=j) = p_j q_r$
for all $1 \leq r \leq N$, $1 \leq j \leq J$,
where $\left(p_1, \ldots, p_J \right)$
are positive values such that $\Sigma_{1 \leq j \leq J} \; p_j = 1$ and $\Lambda(t)$ denotes the leaf where the request $t$ arrives at time $t$. Requests are served according to the IPC rule, \ie, are forwarded upwards until the content is found, and the content is then copied in each cache between this location and the addressed leaf. 

\begin{corol}
Consider a homogeneous tree with \textbf{IPC} scheme and suppose that assumption \textbf{(H)} holds for all its internal nodes. The results of Proposition \ref{APPROXNCRK} then extend to that tree with IRM request process at leaves and Zipf popularity distribution with exponent $\alpha > 1$. 
\end{corol}\label{tree}

\begin{proof}
Only the order of requests in time matters since their precise timing is irrelevant; we can consequently assume that the requests arrive according to a Poisson process with intensity $1$. From the property of independent thinning and merging of Poisson processes, it follows that the requests for a given object $r$ at leaf $j$ is also a Poisson process of intensity 
$p_j q_r$, and that the request process at leaf $j$ is a Poisson process with intensity $p_j$ with a Zipf popularity distribution $q_r = A/r^\alpha$, $r \geq 1$. Now, using assumption \textbf{(H)} and applying the previous results for a single cache to each leaf, we deduce  that at any leaf $j$, the miss sequence for object $r$ is a Poisson process with intensity $p_j q_r M_r^{*}(1)$. Merging these miss sequences from all children of a given second-level node, we deduce that the requests at this node follow a Poisson process and that the probability of request for an object $r$ is  $q_r M_r^{*}(1) / M(1) = q_r(2)$. This process has the same properties as the IRM process with distribution $(q_r(2))_{r \in \mathbb{N}}$ used in the proof of Proposition \ref{APPROXNCRK}, which therefore applies. Repeating recursively this reasoning at each level, we conclude that Proposition \ref{APPROXNCRK} holds in this context
\end{proof}

\begin{remark} Corollary \ref{tree} is also valid for a homogeneous tree where different replacement policies are used at different levels $i$ (e.g. Random at first level and LRU at  second one).
\end{remark}

\section{Numerical results: network of caches}
\label{sec:num_network_cache}
\begin{figure}[!b]
\begin{center}
\subfigure[]{\includegraphics[width=0.45\textwidth]{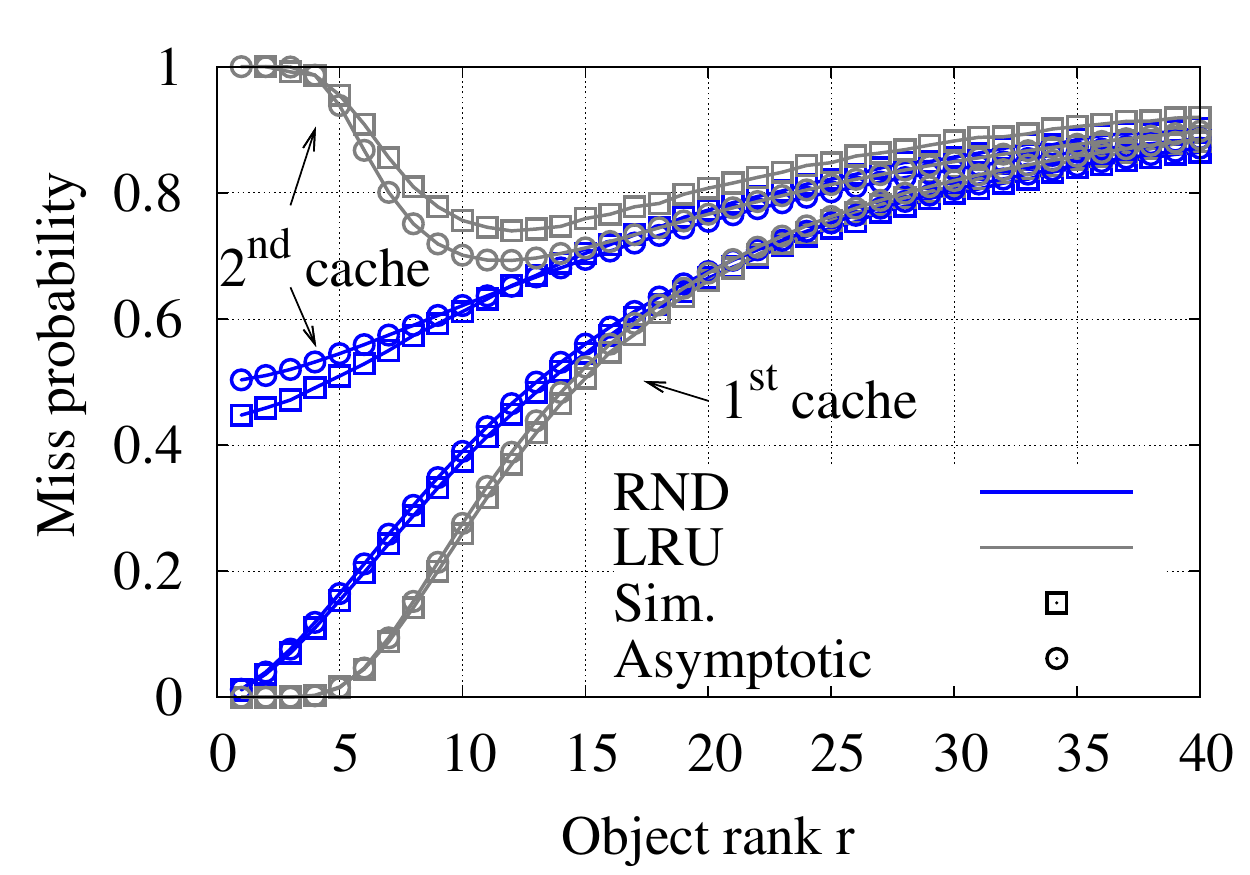}\label{fig:tandem}}
\subfigure[]{\includegraphics[width=0.45\textwidth]{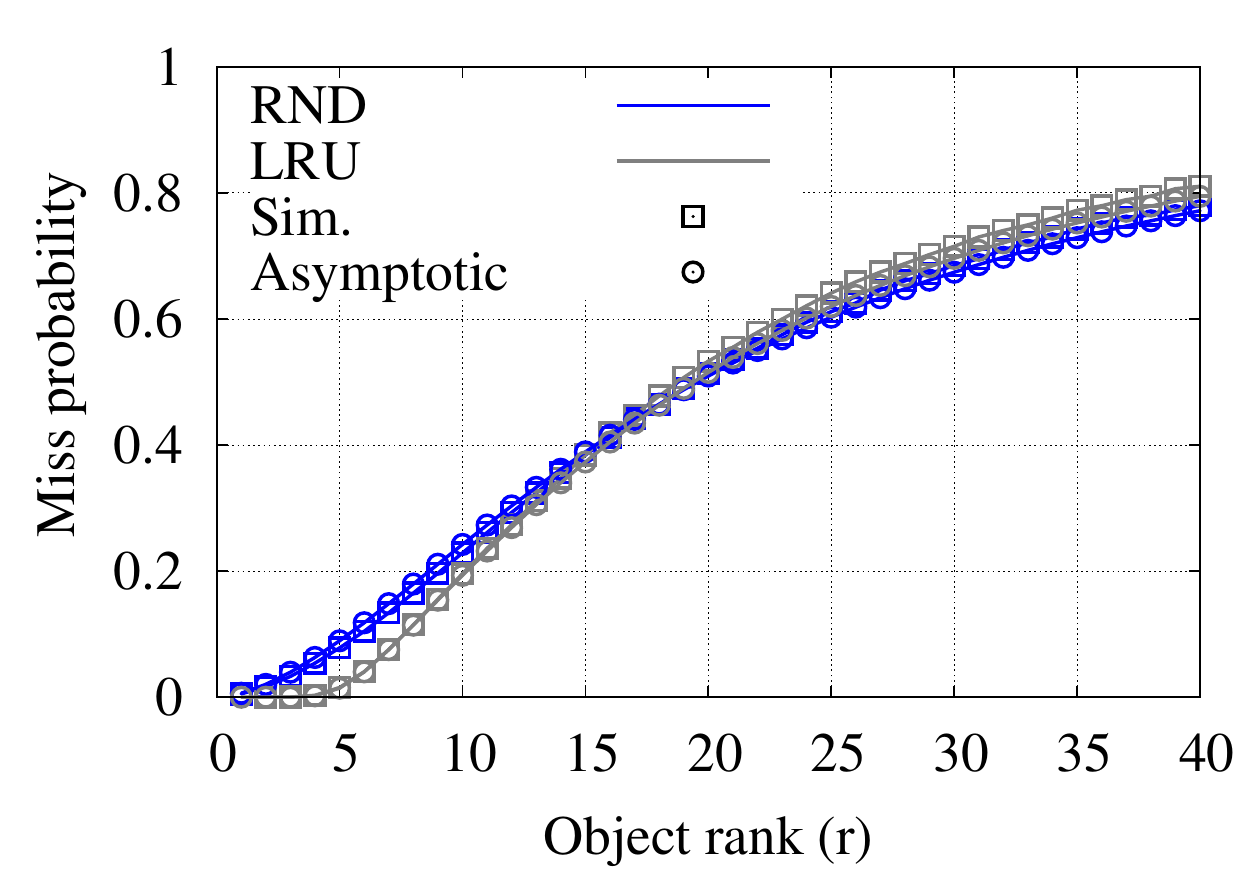}\label{fig:tandem_total_per_class}}
\subfigure[]{\includegraphics[width=0.45\textwidth]{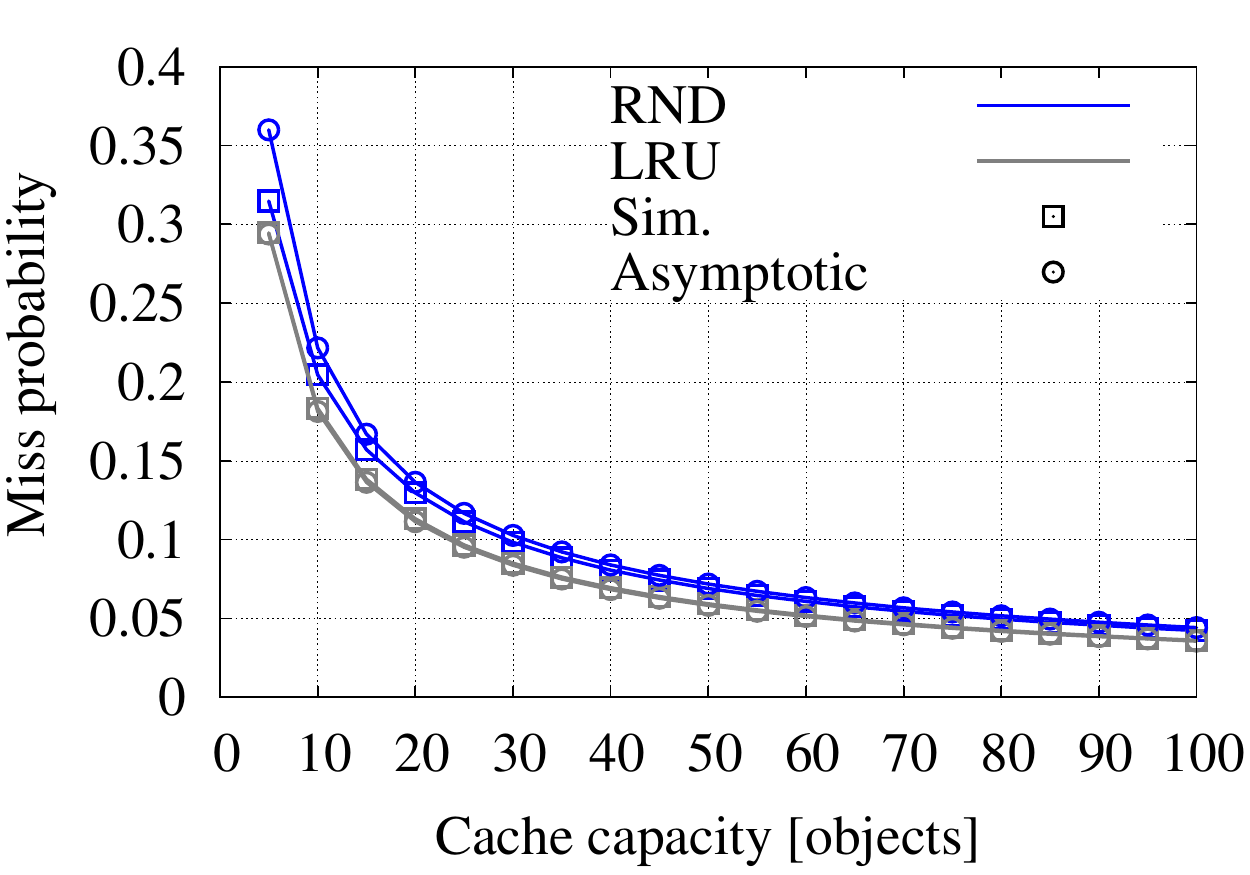}\label{fig:tandem_total}}
\caption{Tandem cache results: (a) asymptotics of $M^*_r(1)$, $M^*_r(2)$ (b) and of   $M_r(2)$ for RND and LRU policies compared to simulation with $C_1=C_2=25$, $\alpha=1.7$ (c) Asymptotic for $M(2)=M(C_1,C_2)$ with $C_1=C_2\leq 100$, $\alpha=1.7$.}
\end{center}
\end{figure} 
We here report numerical and simulation results to show the accuracy of the approximations presented in previous Section \ref{sec:network_cache}.

Fig.~\ref{fig:tandem} first reports estimate (\ref{Miss-r-NCRk})
of $M_r^*(1)$ and  $M_r^*(2)$  for both RND and LRU (the
approximation for the tandem LRU are taken from
\cite{muscaITC2011}) with $C_1=C_2=25$. We focus on the second
cache, as the performance of the first one has been largely
analyzed in previous sections. We note a good agreement between
the approximations evaluated in Section \ref{sec:network_cache}
and  simulation results.  

Moreover, while less popular objects are affected in the same way
when employing either RND or LRU (in our specific example, $r\geq
15$), a significantly different behavior is detectable for
popular objects ($r<15$). Local miss probabilities $M_r^*(1)$ and
$M_r^*(2)$ help understanding where an object has been cached,
conditioned on its rank. The combination of LRU and IPC clearly
tends to favor stationary 
configurations where popular objects are likely to be stored in
the first cache (see also \cite{muscaITC2011} for a similar
discussion). When using RND instead of LRU, however, the
distribution of the content across the two caches is fairly
different; in Fig.~\ref{fig:tandem} for example, while the most
popular objects are likely to be retrieved at the first cache
when using either LRU or RND, only by using RND can such an
object be also found in the second cache. 
It therefore appears that while both LRU and RND tend to store
objects proportionally to their popularity, RND more evenly
distributes objects across the whole path.  

Fig.~\ref{fig:tandem_total_per_class} reports the total miss
probability at the second cache $M_r(2) = M^*_r(1)M^*_r(2) $,
\ie, the probability to query an object of rank $r$ at the
repository server. In this example, we see that objects with rank
$r<15$ are slightly more frequently requested at the server when
using RND rather than LRU, but RND is more favorable than LRU for
objects with higher rank $r\geq 15$. In average, the total miss
probability  at the second cache $M(2)$, reported in
Fig.~\ref{fig:tandem_total}, is very similar either using RND or
LRU, with a slight advantage to LRU. $M(2)$ indicates the amount of data that is to be requested
from the server.  

With no claim of generality, we notice that the approximations calculated in Section \ref{sec:network_cache} for RND and in \cite{muscaICN2011}, \cite{muscaITC2011} for LRU are very accurate. Furthermore, the approximations work well for a large number of tests that we do not report here because of lack of space.
  
\section{Mixture of RND and LRU}
\label{sec:mixed}

So far, we have considered networks of caches where all caches use the RND replacement policy. In practice, it is feasible to use different replacement algorithms in the same network. This section addresses the case of a tandem network, where one cache uses the RND replacement algorithm while the other uses the LRU algorithm. As in Section \ref{TreeSection}, these results also hold in the case of an homogeneous tree.

\subsection{Large cache size estimates}
We first provide estimates for miss probabilities in the case when cache sizes $C_1$ and $C_2$ are large.

\begin{prop}
For the $2$-caches tandem system with \textbf{IPC} scheme, suppose the request process at cache $\sharp 1$ is IRM with Zipf popularity distribution with exponent $\alpha > 1$ and that assumption \textbf{(H)} for cache $\sharp 2$ holds. 
\\
\textbf{I)} When cache $\sharp 1$ (resp. cache $\sharp 2$) uses the RND (resp. LRU) replacement policy, the global (resp. local) miss probability $M_r(2)$ (resp. $M^*_r(2)$) on cache $\sharp 2$ is given by
\begin{align}
M_r(2) & \sim 
\frac{\rhoa\ral}{\displaystyle \rhoa \ral + C_1^\alpha} 
\exp \left ( -\frac{\displaystyle \rhoa C_2^\alpha}
{\alpha \lambda_\alpha \left( \rhoa \ral + C_1^\alpha \right)} \right ),
\nonumber\\
M^*_r(2) & \sim 
\exp \left ( -\frac{\displaystyle \rhoa C_2^\alpha}
{\alpha \lambda_\alpha \left(\rhoa\ral+C_1^\alpha\right)} \right )
\label{MissRandLRU}
\end{align}
for large cache sizes $C_1$, $C_2$ and constants $\rho_\alpha$, $\lambda_\alpha$ introduced in (\ref{prefP}) and (\ref{prefP'}). 
\\
\textbf{II)} When cache $\sharp 1$ (resp. cache $\sharp 2$) uses the LRU (resp. RND) replacement policy, the global (resp. local) miss probability $M_r(2)$ (resp. $M^*_r(2)$) on cache $\sharp 2$ is given by
\begin{align}
M_r(2) & \sim  
\frac{\rho_\alpha r^\alpha}
{\displaystyle \rho_\alpha r^\alpha \exp \left ( \frac{C_1^\alpha}{\alpha \lambda_\alpha r^\alpha } \right ) + C_2^\alpha }, 
\nonumber\\
M^*_r(2) & \sim  
\frac{\rho_\alpha r^\alpha}
{\displaystyle \rho_\alpha r^\alpha + C_2^\alpha
  \exp \left ( -\frac{C_1^\alpha}{\alpha \lambda_\alpha r^\alpha } \right ) }
\label{MissLRURand}
\end{align}
for large cache sizes $C_1$, $C_2$.
\label{PropMixture}
\end{prop}

\begin{proof}
We follow the same derivation pattern as the proof of Proposition \ref{APPROXNCRK} detailed in \textbf{Appendix \ref{ProofApproxNCRk}}. 
\\
\textbf{I)} When cache $\sharp 1$ uses the RND replacement policy, we know from \textbf{Appendix \ref{ProofApproxNCRk}} that the request process at cache $\sharp 2$ is IRM with popularity distribution
$$
q_r(2) = q_r \frac{M_r^*(1)}{M^*(1)} \sim \frac{C_1^{\alpha-1}}{C_1^\alpha + \rhoa \ral}, \; \; r \geq 1,
$$
and is asymptotically Zipf for large $r$. We then follow the proof of Proposition 6.2 of
\cite{muscaITC2011}. Let $S_2(0,t)$ be the number of different
objects requested at cache  $\sharp 2$ in the time interval
$[0,t]$; it verifies 
\[\mathbb{E}\left[S_2(0,t)\right] = \sum_{r \geq 1} \left( 1 -
  e^{-q_r(2) t} \right ). \]
We then first deduce that 
\begin{equation}
\mathbb{E}\left[S_2(0,t)\right] \ge \int_1^{+\infty}
\left(1 -  e^{-q_u(2) t }\right) \mathrm{d}u.
\label{integral}
\end{equation}
Using the variable change $v = C_1^{\alpha-1}t/(C_1^\alpha + \rhoa u^\alpha)$ in the latter integral, we further obtain 
\begin{align*}
\mathbb{E}\left[S_2(0,t)\right] & 
\ge 
\left(
\frac{C_1^{\alpha-1}t}{\rhoa}\right)^{\frac{1}{\alpha}} \times \\
& \int_0^{\frac{C_1^{\alpha-1}t}{C_1^\alpha + \rhoa}}
\frac{1}{\alpha}\left(1-e^{-v}\right)v^{-1-\frac{1}{\alpha}} 
\left(1 - \frac{C_1v}{t}\right)^{\frac{1}{\alpha}-1} \mathrm{d}v
\,.
\end{align*}
Letting $t \uparrow +\infty$, the monotone convergence theorem for function $v \mapsto \left(1-C_1v/t\right)^{-1+1/\alpha}$ together with a further integration by parts yield
\begin{equation}
\lim_{t \uparrow +\infty}
\frac{\mathbb{E}\left[S_2(0,t)\right]^\alpha}{t}
\ge 
\left(
\frac{C_1^{\alpha-1}}{\rhoa}\right)
\left [ \Gamma \left( 1-\frac{1}{\alpha} \right ) \right]^\alpha.
\label{inequ}
\end{equation}
Starting integral (\ref{integral}) from $u=0$ instead of $u=1$, the latter asymptotic bound is seen to hold also as an upper bound of $\mathbb{E}\left[S_2(0,t)\right]^\alpha/t$, thus showing that (\ref{inequ}) actually holds as an equality. The local per-object miss rate on the second cache for an LRU cache is then 
\begin{align*}
M_r^*(2) & \sim \exp\left[-q_r(2) C_2^\alpha \left(\lim_{t \uparrow +\infty}
\frac{\mathbb{E}\left[S_2(0,t)\right]^\alpha}{t}\right)^{-1} \right]
\end{align*}
which proves expressions (\ref{MissRandLRU}).
\\ 
\indent
\textbf{II)} When cache $\sharp 1$ applies the LRU replacement policy, the local per-object miss rate at cache $\sharp 1$ is known \cite{Jelenkovic2008PER} to equal
$$
M_r^*(1) \sim \exp\left[-\frac{C_1^\alpha}{\ral \left[\Gamma\left(1-\frac{1}{\alpha}\right)\right]^\alpha}\right] = \exp \left [ - \frac{C_1^\alpha}{\alpha \lambda_\alpha r^\alpha }\right ]
$$
and that the local average miss rate is 
$$
M^*(1) \sim \frac{1}{\alpha}
\left[\Gamma\left(1-\frac{1}{\alpha}\right)\right]^\alpha
\frac{A}{C_1^{\alpha -1}} = \frac{\lambda_\alpha A}{C_1^{\alpha -1}}.
$$
Using assumption \textbf{(H)}, it then follows that the input process at cache $\sharp 2$ is IRM with popularity distribution given by
\[
q_r(2) = q_r \frac{M_r^*(1)}{M^*(1)} \sim  
\frac{C_1^{\alpha-1}}{\lambda_\alpha r^\alpha} \exp \left[ -\frac{C_1^\alpha}
{\alpha \lambda_\alpha \ral} \right], \; \; r \geq 1.
\]
Note that this distribution is asymptotically Zipf for large $r$,
with coefficient $A'(2) = A/M^*(1)$. Applying estimate
(\ref{asymptmrC2}) to the above defined distribution $q_r(2)$, $r
\geq 1$, it then follows that 
\[M_r^{*}(2) \sim \left ( 1+q_r(2) \theta'(2) \right ) ^{-1}, \]
where the associated root $\theta'(2)$ is easily estimated by $\theta'(2) = C_2^\alpha / A'(2) \rho_\alpha$ by using Lemma \ref{SADDLEPOINTC}. We hence derive that 
$$
\displaystyle M_r^{*}(2) \sim  
\left ( 1 + \displaystyle \frac{C_2^\alpha}{A'(2) \rhoa} \frac{A M_r^*(1)}{\ral M^*(1)} \right ) ^{-1}
$$
which leads to expressions (\ref{MissLRURand})
\end{proof}

\subsection{Numerical results}

\begin{figure}[!b]
\begin{center}
\subfigure[]{\includegraphics[width=0.41\textwidth]{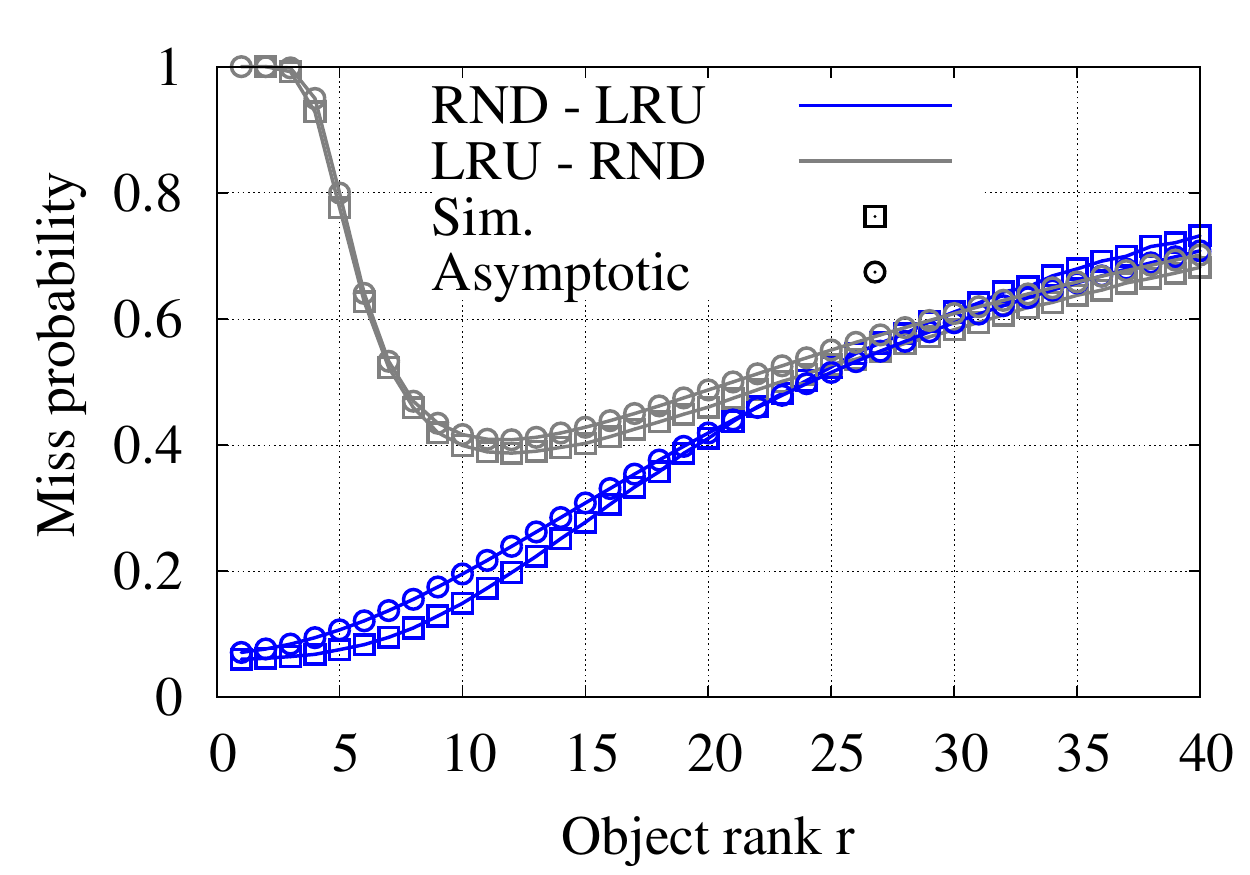}\label{fig:mix_local}}
\subfigure[]{\includegraphics[width=0.41\textwidth]{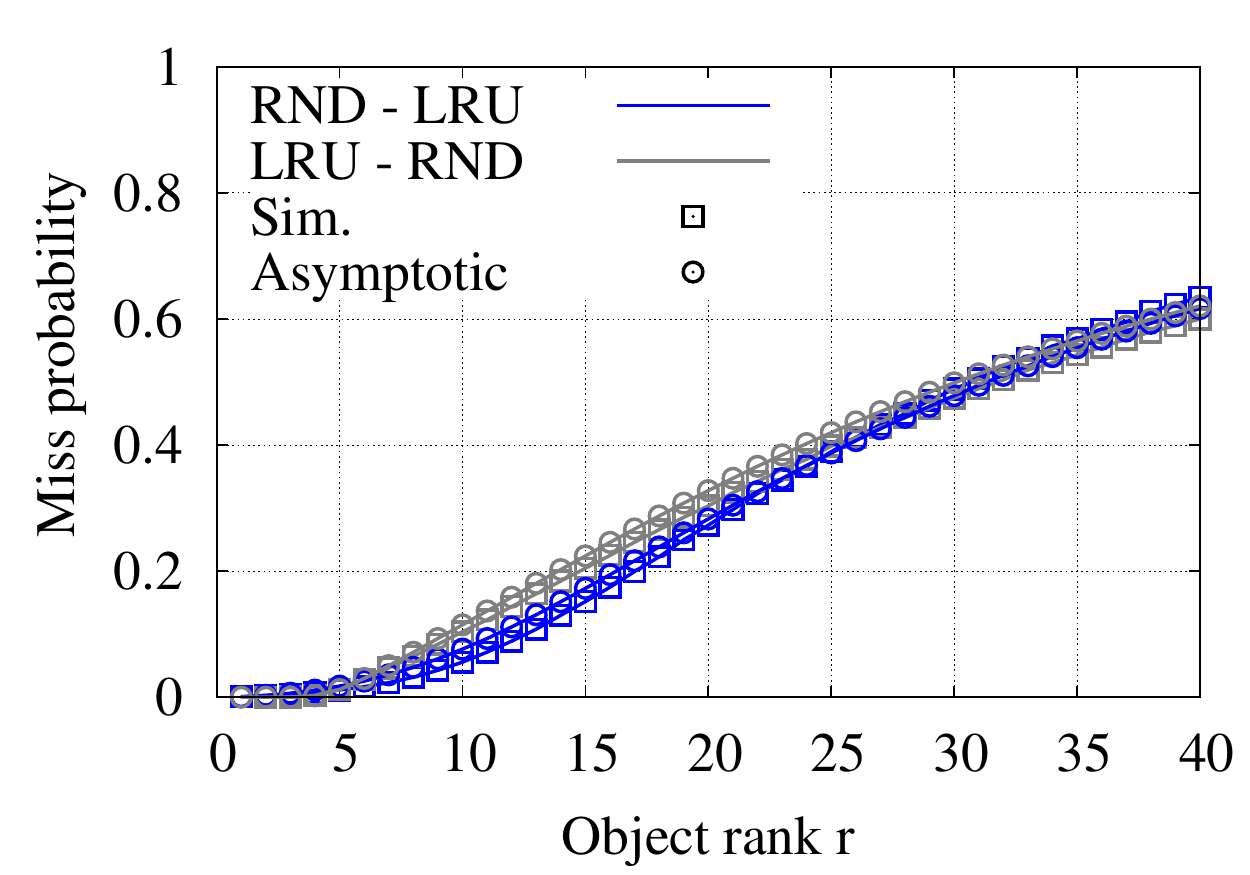}\label{fig:mix_total}}
\subfigure[]{\includegraphics[width=0.41\textwidth]{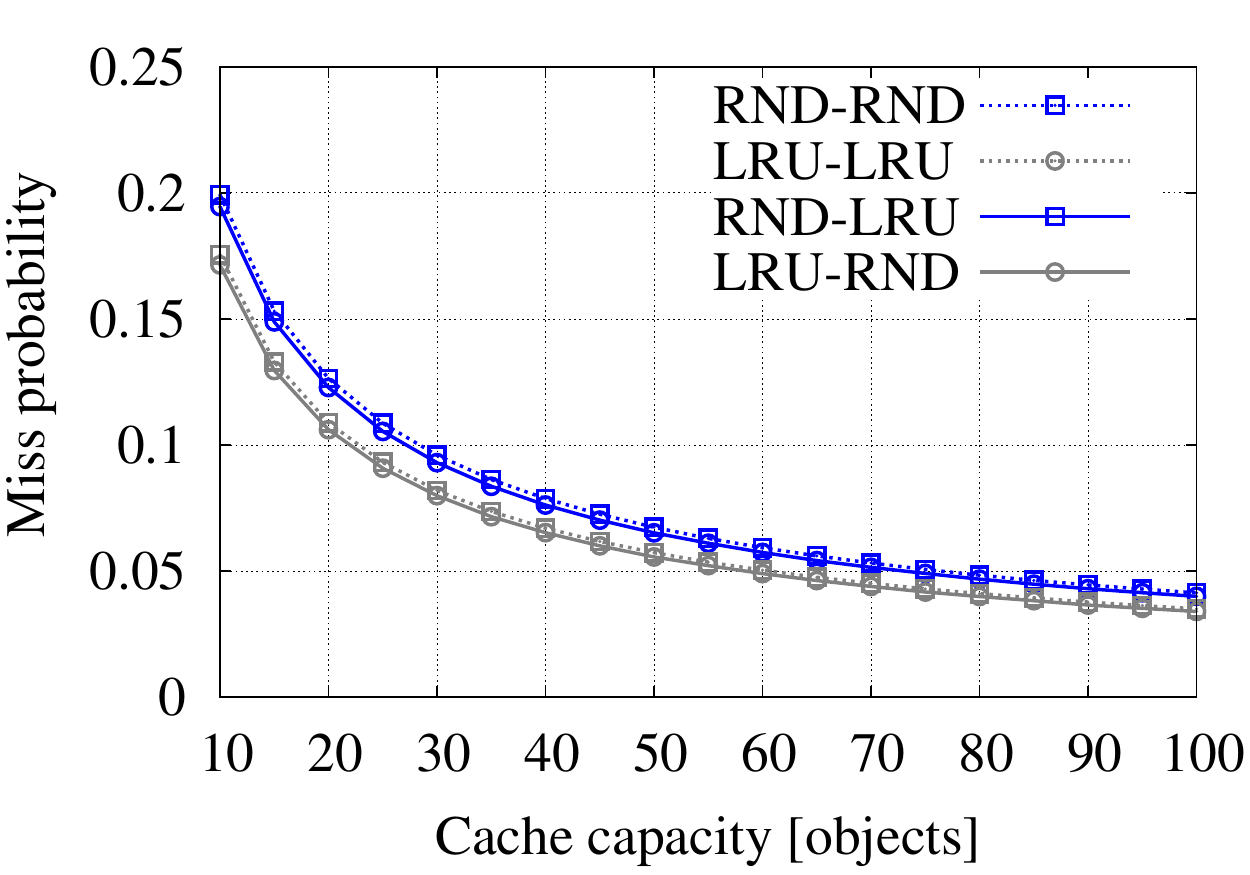}\label{fig:mix_total_C}}
\caption{Mixed homogeneous tree cache results for mixed tree cache networks with $C_1=25, C_2=50$, $\alpha=1.7$: (a) asymptotic of $M^*_r(2)$ (b) asymptotic of $M_r(2)$  (c) Simulation of $M(2)=M(C_1,C_2)$ with $C_1=C_2\leq 100$, $\alpha=1.7$.}
\label{fig:mix}
\end{center}
\end{figure}

We here report numerical and simulation results for mixed
homogeneous tree topologies to show the accuracy of the
approximations presented in previous section, in order to derive
some more general considerations about the mixture of RND and LRU
in a network of caches. According to Section~\ref{TreeSection},
we simulate in fact tree topologies, with 2 leaves with cache
size $C_1$ and one root with cache size $C_2$.

Fig.~\ref{fig:mix_local} reports $M^*_r(2)$ for RND-LRU and
LRU-RND homogeneous tree networks with asymptotics
(\ref{MissRandLRU}) and (\ref{MissLRURand}), respectively. We
first note that the latter provide estimates with reasonable
accuracy. Besides, we observe that the behavior of $M^*_r(2)$ is
in strict relation to the policy used for cache $\sharp 1$. If
cache $\sharp 1$ is RND then, $M^*_r(2)$ has a behavior similar
to that observed in RND caches tandem; similarly, if cache
$\sharp 1$ is LRU, $M^*_r(2)$ behaves as in the case of LRU
caches in tandem.  

This phenomenon has a natural explanation. RND and LRU act in the
same way on objects ranked in the tail of the Zipf popularity
distribution. However, the two replacement policies manage
popular objects in a rather different way, as we already observed
in Section \ref{sec:num_network_cache}. The second level caches see a local popularity that is shaped, by the first level of caches, in the body of the probability distribution.
The portion of the distribution that is affected by such shaping process is determined by the cache size $C_1$ at first level. In the analysis reported in Fig.~\ref{fig:mix}, the first level cache significantly determines the performance of the overall tandem system. In Fig.~\ref{fig:mix_total}, we observe that the total miss probability $M_r(2)$ in the two mixed tandem caches is similar, while the distribution of the objects across the two nodes varies considerably. 

Finally, in Fig.~\ref{fig:mix_total_C} we present simulation results of $M(2)=M(C_1,C_2)$ for all possible configurations in an homogeneous tree network in which $C_1=C_2\leq 100$. Observe that LRU-RND tree cache network achieve slightly better performance than the LRU-LRU system at least with Zipf shape parameter $\alpha = 1.7$. This behaviour would suggest to prefer LRU at the first cache beacuse it performs better in terms of Miss probability and RND at the second one in order to save significant processing time.

\section{Conclusion}\label{sec:conclusion}

The recent technological evolution of memory capacities, as illustrated by the deployment of CDNs and the proposition of new information-centric architectures where caching becomes an intrinsic network property, raise new interests on caching studies.

In this paper, we have studied the RND replacement policy, where
objects to be removed are chosen uniformly at random. Assuming
that the content popularity follows a Zipf law with parameter
$\alpha > 1$, and that content requests are IRM, we prove that
for large cache size $C$, the miss rate is asymptotically
equivalent to $A \rho_{\alpha} C^{1-\alpha}$ (see Lemma
\ref{ASYMPTF} and Proposition \ref{ASYMPT}). This shows that the
difference between LRU and RND caches is independent of the cache
size and only depends on coefficient $A\rho_{\alpha}$. These
results  
are extended to typical network topologies, namely tandems and
homogeneous trees, under the assumption that requests are IRM at
any node. The case of mixed policies (RND at one network level
and LRU at the other one) is also considered. Simulations show
that the IRM assumption applied to network topologies is
reasonable and that provided estimates are accurate. 

Our results suggest that the performance of RND is reasonably close to that of LRU. As a consequence, RND is a good candidate for high-speed caching when the complexity of the replacement policy becomes critical. In the presence of a hierarchy of caches, caches at deep levels (i.e. access networks) typically serve a relatively small number of requests per second which can be easily sustained by a cache running LRU; LRU policy should consequently be implemented at the bottom level since it provides the best performance. Meanwhile, higher-level caches see many aggregated requests and should therefore use the RND policy which yields similar performance (as to second-level cache) while being less computationally expansive. 

We have assumed in this paper that the parameter $\alpha$ of the Zipf distribution is larger than $1$, and the total number of available objects is infinite. As an object of further study, we first intend to explore the case where $\alpha \leq 1$, with a finite number of objects. Besides, since Zipf popularity distributions do not represent all types of Internet traffic, we also intend to analyze the performance of RND caches when the content requests follow a light-tailed (e.g. Weibull) distribution. Finally, all results derived in this paper hold for i.i.d. content requests. Admittedly, actual traces show that requests are correlated, both in time and space. The definition of an accurate and realistic model which can take these correlations into account, as well as the extension of the present results to such a request model is also on our research agenda. 

\section*{Acknowledgements}
This work has been partially funded by the French National Research Agency (ANR), 
CONNECT project, under grant number ANR-10-VERS-001, and by the European FP7 IP
project SAIL under grant number 257448. We thank Philippe Olivier
for his useful comments. 

\bibliographystyle{abbrv}
\bibliography{sigproc}

\appendix

\setcounter{tocdepth}{-1}

\section{Proof of Theorem \ref{LD}}
\label{LargeDeviations}

\indent \textit{(i)} Using (\ref{genf}), equation (\ref{saddlepoint0}) reduces to $g(z) = C$ where
\begin{equation}
g(z) = z \frac{F'(z)}{F(z)} = \sum_{j \geq 1} \frac{q_j z}{1 + q_j z}.
\label{saddlepoint}
\end{equation}
Continuous function $g:z \in [0,+\infty[ \; \rightarrow g(z) \in [0,+\infty[$ vanishes at $z = 0$, is strictly increasing on $[0, +\infty[$ and tends to $+\infty$ when $z \uparrow +\infty$. Equation (\ref{saddlepoint0}) has consequently a unique positive solution $\theta_C$. Note that $\theta_C$ tends to $+\infty$ with $C$ since $g(z) \leq \Sigma_{j \geq 1} (q_j z) = z$, hence $\theta_C \geq C$. 
\\
\indent \textit{(ii)} Consider the random variable $X_C$ with distribution
\begin{equation}
\mathbb{P}(X_C = x) = \frac{G(x)}{F(\theta_C)}\theta_C^x, \; \; x \geq 0,
\label{distY}
\end{equation}
where $\theta_C$ satisfies (\ref{saddlepoint0}); note that definition (\ref{distY}) for $X_C$ is equivalent to
\begin{equation}
G(x) = \frac{F(\theta_C)}{\theta_C^x}\mathbb{P}(X_C = x), \; \; x \geq 0.
\label{distX}
\end{equation}
Using definition (\ref{distY}), the generating function of random variable $X_C$ is  $z \mapsto F(z\theta_C)/F(\theta_C)$; in view of (\ref{saddlepoint0}), the expectation of variable $X_C$ is then
$$
\mathbb{E}(X_C) = \frac{\mathrm{d}}{\mathrm{d}z} \frac{F(z \theta_C)}{F(\theta_C)}\Big |  _{z = 1} = \theta_C \frac{F'(\theta_C)}{F(\theta_C)} = C
$$
so that random variables $Y_C = (X_C - C)/\sqrt{C}$, $C \geq 0$, are all centered. Besides, the Laplace transform of variable $Y_C$ is given by
$$
\mathbb{E}(e^{-sY_C}) = e^{s\sqrt{C}} \mathbb{E}(e^{-sX_C/\sqrt{C}}) = e^{s\sqrt{C}}\frac{F(\theta_C e^{-s/\sqrt{C}})}{F(\theta_C)}
$$
for all $s \in \mathbb{C}$. By L\'evy's continuity theorem (\cite{Rao}, Theorem 4.2.4), assumption (\ref{Gaussapprox1}) entails that variables $Y_C$ converge in distribution when $C \uparrow +\infty$ towards a centered Gaussian variable with variance $\sigma^2$; moreover, assumption (\ref{Gaussapprox2}) ensures that the conditions of Chaganty-Sethuraman's theorem (\cite{ChaganSeth}, Theorem 4.1) hold so that  $\mathbb{P}(X_C = C) = \mathbb{P}(Y_C = 0)$ is asymptotic to
\begin{equation}
\mathbb{P}(X_C = C) \sim 1/\sigma \sqrt{2\pi C}
\label{Gaussapprox0}
\end{equation}
as $C \uparrow +\infty$. Equation (\ref{distX}) for $x = C$ and asymptotic  (\ref{Gaussapprox0}) together provide estimate (\ref{Saddlepointfin}) for $G(C)$
\qed

\section{Proof of Lemma \ref{ASYMPTF}}
\label{ProofAsymptF}

Let $q(x) = A/x^\alpha$ and $f_z(x) = \log(1 + q(x)z)$ for any real $x \geq 1$ and $z \in \mathbb{C}$; definitions (\ref{genf}), (\ref{Zipf}) and the above notation then entail that 
$$
\log F(z) = \sum_{r \geq 1} f_z(r);
$$
function $\log F$ is analytic in the domain $\mathbb{C} \setminus \mathbb{R}^-$. For given $z \in \mathbb{C} \setminus \mathbb{R}^-$ and integer $J \geq 1$, the Euler-Maclaurin summation formula \cite{Abram} reads
\begin{align}
\sum_{r = 1}^J f_z(r) = & \int_{1}^{J} f_z(x) \mathrm{d}x + \frac{1}{2} \left [ f_z(J) +  f_z(1) \right ] +
\nonumber \\
& \frac{1}{12} \left [ f_z'(J) - f_z'(1) \right ] + \frac{T_z(J)}{6}
\label{EulerMac}
\end{align}
with
$$
T_z(J) = \int_{1}^J B_3(\{x\})f_z^{(3)}(x) \mathrm{d}x,
$$
where $B_3(x) = x(x-1)(2x-1)/2$ is the third Bernoulli polynomial and $\{x\}$ denotes the fractional part of real $x$; derivatives of $f_z$ are taken with respect to $x$. Consider the behavior of the r.h.s. of (\ref{EulerMac}) as $J$ tends to infinity. We first have $f_z(1) = \log(1+Az)$ and $f_z(J) = O(J^{-\alpha})$ for large $J$; differentiation entails
$$
f_z'(1) = - \frac{\alpha Az}{1 + Az}
$$
and $f_z'(J) = O(J^{-\alpha - 1})$ for large $J$. Differentiating twice again with respect to $x$ shows that the third derivative of $f_z$ is $O(x^{-\alpha - 3})$ for large positive $x$, and is consequently integrable at infinity. Letting $J$ tend to infinity in (\ref{EulerMac}) and using the above observations together with the boundedness of periodic  function $x \geq 1 \mapsto B_3(\{x\})$, we obtain
\begin{equation}
\log F(z) = \int_{1}^{+\infty} f_z(x) \mathrm{d}x + \frac{1}{2} \log(1+Az) + \frac{\alpha Az}{12(1 + Az)} + \frac{T_z}{6}
\label{EulerMac1}
\end{equation}
where
\begin{equation}
T_z = \int_{1}^{+\infty} B_3(\{x\}) f_z^{(3)}(x) \mathrm{d}x.
\label{Tz}
\end{equation}
Now, considering the first integral in the r.h.s. of (\ref{EulerMac1}), the variable change $x = t(Az)^{1/\alpha}$ gives
\begin{align}
\int_{1}^{+\infty} f_z(x) \mathrm{d}x & = (Az)^{1/\alpha} \Big [ L - \int_0^{1/(Az)^{1/\alpha}} \log \left ( 1 + \frac{1}{t^\alpha} \right ) \mathrm{d}t \Big ]
\nonumber \\
& = L(Az)^{1/\alpha} - \log(Az) - \alpha + o(1)
\label{EulerMac2}
\end{align}
where $L$ is the finite integral \cite{Abram}
\begin{equation} 
L = \int_{0}^{+\infty} \log \left ( 1 + \frac{1}{t^\alpha} \right ) \mathrm{d}t 
= \frac{\pi}{\sin(\pi/\alpha)} = \alpha \rho_\alpha^{1/\alpha}
\label{intL}
\end{equation}
with $\rho_\alpha$ introduced in (\ref{prefP}) for $\alpha > 1$, and where
\begin{equation}
\int_0^{1/(Az)^{1/\alpha}} \log \left ( 1 + \frac{1}{t^\alpha }\right ) \mathrm{d}t = \frac{\log(Az)}{(Az)^{1/\alpha}} + \frac{\alpha}{(Az)^{1/\alpha}} + o(1).
\label{intLbis}
\end{equation}
Gathering (\ref{intL}) and (\ref{intLbis}) provides expansion (\ref{EulerMac2}). Using the explicit expression of $f^{(3)}_z$, the dominated convergence theorem finally shows that when $z \uparrow +\infty$, remainder $T_z$ in (\ref{Tz}) tends to some finite constant $t_\alpha$ depending on $\alpha$ only. Gathering terms in 
(\ref{EulerMac1})-(\ref{EulerMac2}), we are finally left with expansion (\ref{asympt}) with constant $S_\alpha = -\alpha + \alpha/12 + t_\alpha/6$. Some further calculations would provide $S_\alpha = -\alpha \log(2\pi)/2$, although this actual value does not intervene in our discussion
\qed

\section{Proof of Proposition \ref{SADDLEPOINTC}}
\label{ProofSaddlepointC}

Recall definition (\ref{saddlepoint}) of function $g$ and write equivalently
$$
g(z) = \sum_{r \geq 1} g_z(r).
$$
where we let $g_z(x) = Az(x^\alpha + Az)^{-1}$. The Euler-Maclaurin summation formula \cite{Abram} applies again in the form
\begin{align}
\sum_{r = 1}^J g_z(r) = & \int_{1}^{J} g_z(x) \mathrm{d}x + \frac{1}{2} \left [ g_z(J) +  g_z(1) \right ] + 
\nonumber \\
& \frac{1}{12} \left [ g_z'(J) - g_z'(1) \right ] + \frac{W_z(J)}{6}
\label{EulerMacbis}
\end{align}
for given $z\in \mathbb{C} \smallsetminus \mathbb{R}^-$, integer $J \geq 1$ and where 
$$
\mid W_z(J) \mid \; \leq \frac{12}{(2\pi)^{2}} \int_{1}^J \mid g_z^{(3)}(x) \mid \mathrm{d}x
$$
(derivatives of $g_z$ are taken with respect to variable $x$). Consider the behavior of the r.h.s. of (\ref{EulerMacbis}) as $J$ tends to infinity. Firstly, $g_z(1) = Az/(Az+1)$ and $g_z(J) = O(J^{-\alpha})$ as $J \uparrow +\infty$; secondly, $g_z'(1) = -A\alpha z (1+Az)^{-2}$ together with $g_z'(J) = O(J^{-\alpha - 1})$ for large $J$. Differentiating twice again shows that the third derivative of $g_z$ is $O(x^{-\alpha - 3})$ for large positive $x$ and is consequently integrable at infinity. Letting $J$ tend to infinity in (\ref{EulerMacbis}) therefore implies equality
\begin{equation}
g(z) = \int_{1}^{+\infty} g_z(x) \mathrm{d}x + \frac{1}{2}\frac{Az}{Az+1} +  \frac{1}{12}\frac{A\alpha z}{(1+Az)^2} + \frac{W_z}{6}
\label{EulerMacbis1}
\end{equation}
where
$$
\mid W_z \mid \; \leq \frac{12}{(2\pi)^{2}} \int_{1}^{+\infty} \mid g_z^{(3)}(x) \mid \mathrm{d}x.
$$
Using the explicit expression of the derivative $g_z^{(3)}$, it can be simply shown that  
\begin{equation}
\mid W_z \mid \; = O(z^{-2/\alpha}), \; \mid W_z \mid \; = O(z^{-1}\log z), \; \mid W_z \mid \; = O(z^{-1})
\label{Wz}
\end{equation}
if $\alpha > 2$, $\alpha = 2$ and $1 < \alpha < 2$, respectively. Now, considering the first integral in the r.h.s. of (\ref{EulerMacbis1}), the variable change $x = t(Az)^{1/\alpha}$ gives
\begin{equation}
\int_{1}^{+\infty} g_z(x) \mathrm{d}x = I(Az)^{1/\alpha} - 1 + O \left( \frac{1}{z} \right )
\label{EulerMacbis2}
\end{equation}
where $I = L/\alpha = \rho_\alpha^{1/\alpha}$, with integral $L$ introduced in (\ref{intL}) for $\alpha > 1$. Expanding all terms in powers of $z$ for large $z$, it therefore follows from (\ref{EulerMacbis1}) and (\ref{EulerMacbis2}) that
$$
g(z) = \rho_\alpha^{1/\alpha}(Az)^{1/\alpha} - \frac{1}{2} + W_z 
$$
with $W_z$ estimated in (\ref{Wz}). For large $C$, equation (\ref{saddlepoint}), \ie $g(\theta_C) = C$, then reads
\begin{align}
A\theta_C & = \left [ \frac{C}{I} + \frac{1}{2I} + O(W_{\theta_C}) \right ]^\alpha 
\nonumber \\
& = \left ( \frac{C}{I} \right )^\alpha + \frac{\alpha}{2I^\alpha}C^{\alpha - 1} + O(C^{\alpha - 2})
\label{thetaC1}
\end{align} 
for $\alpha > 2$ since (\ref{Wz}) implies $W_{\theta_C} = O(\theta_C^{-2/\alpha}) = O(C^{-2})$ in this case. The case $1 < \alpha < 2$ gives a similar expansion since the remainder is $W_{\theta_C}/C = O(C^{-\alpha}/C) = O(C^{-\alpha - 1})$. Finally, the case $\alpha = 2$ yields
\begin{equation}
A\theta_C = \left [ \frac{C}{I} + \frac{1}{2I} + O(W_{\theta_C}) \right ]^2 = \left ( \frac{C}{I} \right )^2 + O \left ( \frac{\log C}{C} \right ).
\label{thetaC2}
\end{equation} 
Gathering results (\ref{thetaC1})-(\ref{thetaC2}) finally provides expansions (\ref{zC}) for $\theta_C$ 
\qed

\section{Proof of Proposition \ref{ASYMPT}}
\label{ProofAsympt}

We here verify that conditions (\ref{Gaussapprox1}) and (\ref{Gaussapprox2}) of Theorem \ref{LD} are satisfied in the case of a Zipf popularity distribution with exponent $\alpha > 1$. Let us first establish convergence result (\ref{Gaussapprox1}). Using Lemma \ref{ASYMPTF}, we readily calculate
\begin{align}
& e^{s\sqrt{C}} \frac{F(\theta_C e^{-s/\sqrt{C}})}{F(\theta_C)} = \; e^{s\sqrt{C}} \exp \Big [ \alpha(\rho_\alpha A\theta_C)^{1/\alpha} \times  
\nonumber \\
& \left (e^{-s/\alpha \sqrt{C}} - 1 \right ) + \frac{s}{2 \sqrt{C}} + \varepsilon(\theta_C e^{-s/\sqrt{C}}) - \varepsilon(\theta_C) \Big ]
\label{CV}
\end{align}
for any given $s \in \mathbb{C}$ with $\Re(s) = 0$, $\mid \Im(s) \mid \; \leq a$ and where $\varepsilon(\theta) \rightarrow 0$ as $\theta \uparrow +\infty$. By Lemma \ref{SADDLEPOINTC}, we further obtain $\alpha(\rho_\alpha A\theta_C)^{1/\alpha} = \alpha C + \alpha \rho_\alpha r_C + o(r_C)$ and the expansion of $e^{-s/\alpha \sqrt{C}} - 1$ at first order in $1/C$ entails that
$$
\alpha (\rho_\alpha A\theta_C)^{1/\alpha} \left (e^{-s/\alpha \sqrt{C}} - 1 \right ) = -s \sqrt{C} + \frac{s^2}{2\alpha} + O \left ( \frac{1}{\sqrt{C}} \right );
$$
letting $C$ tend to infinity, we then derive from (\ref{CV}) and the previous expansions that
$$
e^{s\sqrt{C}} \frac{F(\theta_C e^{-s/\sqrt{C}})}{F(\theta_C)} \rightarrow \exp \left ( \frac{s^2}{2\alpha} \right )
$$
so that assumption (\ref{Gaussapprox1}) is satisfied with $\sigma^2 = 1/\alpha$.
\\
\indent
Let us finally verify boundedness condition (\ref{Gaussapprox2}). Rephrasing (\ref{CV}) for $s = -iy \sqrt{C}$, we have
\begin{align}
\left (  \frac{F(\theta_C e^{iy})}{F(\theta_C)} \right ) ^{1/C} = & \exp \Big [ \frac{\alpha (\rho_\alpha A\theta_C)^{1/\alpha}}{C} \left (e^{iy/\alpha} - 1 \right ) - \frac{iy}{2 C} 
\nonumber \\
& + \frac{\varepsilon(\theta_C e^{iy}) - \varepsilon(\theta_C)}{C} \Big ]
\end{align}
for any $y \in \mathbb{R}$. But as above, $\alpha (\rho_\alpha A\theta_C)^{1/\alpha}/C$ tends to the constant $\alpha$ when $C \uparrow +\infty$ so that
\begin{equation}
\left | \frac{F(\theta_C e^{iy})}{F(\theta_C)} \right | ^{1/C} \leq \mid h(y) \mid ^\beta \times \Big | \exp \Big [ \frac{\varepsilon(\theta_C e^{iy}) - \varepsilon(\theta_C)}{C} \Big ] \Big |
\label{CVbis}
\end{equation}
for some positive constant $\beta$ and where
\balancecolumns
$$
h(y) = \left | \exp \left ( e^{iy/\alpha} - 1 \right ) \right | = \exp \left ( \cos \left (\frac{y}{\alpha} \right ) - 1 \right ).
$$

Function $h$ is continuous, even and given $\delta > 0$, $h$ is decreasing on interval $[\delta, \pi]$ since $\alpha > 1$, hence $h(y) \leq h(\delta) = \eta_\delta < 1$ for $\delta \leq y \leq \pi$. Using the estimates derived in Appendix \ref{ProofAsymptF}, it is further verified that, given any compact $\mathcal{K} \subset \mathbb{C}$ not containing the origin, we have $\lim_{C \uparrow +\infty} \varepsilon(\theta_C u) = 0$ uniformly with respect to $u \in \mathcal{K}$; this entails that the exponential term in the right-hand side of (\ref{CVbis}) tends to 1 when $C \uparrow +\infty$ uniformly with respect to $u = e^{iy}$, $y \in [\delta,\pi]$. We finally conclude that condition (\ref{Gaussapprox2}) is verified
\qed

\section{Proof of Proposition \ref{APPROXNCRK}}
\label{ProofApproxNCRk}

Let $q_r(\ell + 1)$, $r \geq 1$, denote the distribution of the input process at cache $\sharp (\ell+1)$, $\ell \geq 1$. By the same reasoning than that performed in Lemma \ref{GENr}, we can write
\begin{equation}
q_r(\ell+1) = q_r(\ell) \frac{M^*_r(\ell)}{M^*(\ell)} = q_r \frac{M^*_r(\ell)...M^*_r(1)}{M^*(\ell)...M^*(1)}
\label{NCRk1}
\end{equation}
for all $r \in \mathbb{N}$, where $M^*_r(\ell)$ (resp. $M^*(\ell)$) is the local miss probability of a request for object $r$ at cache $\sharp \ell$ (resp. the averaged local miss probability for all objects requested at cache $\sharp \ell$) introduced in (\ref{averagelocal}) and with notation $q_r = q_r(1)$. As $M_r^*(\ell') \rightarrow 1$ for all $\ell' \leq \ell$ when  $r \uparrow +\infty$, we deduce from (\ref{NCRk1}) that
$$
q_r(\ell+1) \sim \frac{A(\ell)}{r^\alpha}
$$
when $r \uparrow +\infty$, where $A(\ell) = A/M^*(1)M^*(2)...M^*(\ell)$. Apply then estimate (\ref{asymptmrC2}) to obtain
\begin{equation}
M^*_r(\ell+1) \sim \frac{1}{1 + \theta(\ell + 1) q_r(\ell+1)}
\label{NCRk3}
\end{equation}
where $\theta(\ell + 1) \sim C_{\ell+1}^\alpha/A(\ell) \rho_\alpha$
and with $q_r(\ell+1)$ given by (\ref{NCRk1}); using the value of $A(\ell)$ above and the definition $q_r = A/r^\alpha$, the product $\theta(\ell + 1) q_r(\ell+1)$ reduces to 
\begin{align}
\theta(\ell + 1) q_r(\ell+1) & \sim \frac{C_{\ell+1}^\alpha}{A(\ell) \rho_\alpha} \cdot q_r \frac{M^*_r(\ell)...M^*_r(1)}{M^*(\ell)...M^*(1)} 
\nonumber \\
& = \frac{C_{\ell+1}^\alpha}{\rho_\alpha r^\alpha}M^*_r(\ell)...M^*_r(1).
\label{NCRk4}
\end{align}
Writing then $M_r(\ell+1) = M_r(\ell) M^*_r(\ell+1)$, asymptotics (\ref{NCRk3}) and (\ref{NCRk4}) together yield
$$
M_r(\ell+1) \sim M_r(\ell) \left [ 1 + \frac{C_{\ell+1}^\alpha}{\rho_\alpha r^\alpha}M^*_r(\ell)...M^*_r(1) \right ]^{-1}
$$
so that
$$
\frac{1}{M_r(\ell+1)} \sim \frac{1}{M_r(\ell)} + \frac{C_{\ell+1}^\alpha}{\rho_\alpha r^\alpha}
$$
since $\Pi_{j=1}^\ell M^*_r(j) = M_r(\ell)$ after (\ref{independence}); the latter recursion readily provides expression (\ref{Miss-r-NCRk}) for $M_r(\ell)$, $1 \leq \ell \leq K$. 
\\
Using relation $M_r(\ell+1) = M_r(\ell) M^*_r(\ell+1)$ again together with expression (\ref{Miss-r-NCRk}) of $M_r(\ell)$ provides in turn expression (\ref{Miss-r-NCRk})
for $M^*_r(\ell)$, $1 \leq \ell \leq K$  
\qed

\end{document}